\documentclass[12pt, reqno]{amsart}
\usepackage{xcolor}
\usepackage [latin1]{inputenc}
\usepackage{graphicx}
\usepackage{amssymb}
\usepackage{color}
\usepackage{epstopdf}
\usepackage{amsthm,amsmath,amssymb}
\usepackage{mathbbol}
\usepackage{mathrsfs}
\usepackage{cite}
\usepackage{subeqnarray}
\usepackage{cases}
\usepackage{titletoc}
\usepackage{breqn}
\usepackage{amsmath}
\usepackage{caption}
\allowdisplaybreaks[4]

\makeatletter

\newcommand{\Rmnum}[1]{\expandafter\@slowromancap\romannumeral #1@}
\makeatother

\newtheorem{proposition}{Proposition}[section]

\newtheorem{lemma}[proposition]{Lemma}
\newtheorem{remark}[proposition]{Remark}
\newtheorem{thm}[proposition]{Theorem}
\newtheorem{definition}[proposition]{Definition}
\newtheorem{ex}[proposition]{Example}

\begin{document}

\begin{center}
\title{\large \sc \bf On a matrix constrained CKP hierarchy}

\maketitle

\author{Song Li}

\address{School of Mathematics (Zhuhai)\\ Sun Yat-sen University\\ Zhuhai, Guangdong, 519082, China.\\Email: lisong57@mail2.sysu.edu.cn}

\medskip
\author{Kelei Tian }

\address{School of Mathematics\\ Hefei University of Technology\\ Hefei, Anhui, 230601, China.\\Email: kltian@hfut.edu.cn}

\medskip
\author{Zhiwei Wu}

\address{Corresponding author\\ School of Mathematics (Zhuhai)\\ Sun Yat-sen University\\ Zhuhai, Guangdong, 519082, China.\\Email: wuzhiwei3@mail.sysu.edu.cn}
\bigskip

\renewcommand{\thefootnote}{\fnsymbol{footnote}}


\end{center}

\bigskip
\textbf{Abstract:} The algebraic structures of integrable hierarchies play an important role in the study of soliton equations. In this paper, we use splitting theory to give a matrix  representation of a constrained CKP hierarchy, which can be considered as a generalization
of the $\hat{A}_{2n}^{(2)}$-KdV hierarchy and the constrained KP hierarchy. An equivalent construction in terms of the pseudo-differential operator is discussed. Darboux transformations, scaling transformation and tau functions $\ln \tau_f$ for this constrained hierarchy are studied. Moreover, we present formulas for the Virasoro vector fields on $\ln \tau_f$ for the $\hat{A}_{2 n}^{(2)}$-KdV hierarchy.

\medskip


\bigskip


\section{\sc \bf Introduction}

Integrable equations, which describe nonlinear phenomena ubiquitous in domains such as fluid dynamics, plasma physics, biological systems, quantum field theory and etc, are essential in \cite{DavydovSi1991,HadaMNw1997,LetourneuxVQg1992, Mineev-WeinsteinIs2000}. In literature, one systematic way to generate such kind of equations is by the pseudo-differential operators \cite{DickeySe2003}.

Let $\mathcal{D}$ be the algebra of pseudo-differential operators $X=\sum_{i\leq i_{0}}X_i\partial^i, X_i\in C^{\infty}(\mathbb{R}, \mathbb{C})$.
Let symbols $( A)_{\geq 0}$ and $( A)_{<0}$  denote $\sum\nolimits_{i=0}^{m} a_{i}\partial^{i}$ and $\sum\nolimits_{i=-\infty}^{-1} a_{i}\partial^{i}$ respectively for $A=\sum\nolimits_{i=-\infty}^{m} a_{i}\partial^{i}\in\mathcal{D}$.
Let $\mathcal{N}$ denote the algebra of pseudo-differential operators of the form
\begin{eqnarray}
L=\partial+u_1\partial^{-1}+u_2\partial^{-2}+\cdots, \nonumber
\end{eqnarray}
where $\partial=\frac{\partial}{\partial x}$ and the coefficients $u_i=u_i(x, t)$ are functions.
It was shown that the following equation generated hierarchies of commuting Hamiltonian flows on $\mathcal{N}$,
\begin{eqnarray}
	\frac{\partial L}{\partial t_{n}}=[(L^{n})_{\geq 0},L],\;\;n=1,2,3,\cdots,\nonumber
\end{eqnarray}
which is  the KP (Kadomtsev-Petviashvili) hierarchy \cite{DateJKMNi1983,JimboMSa1983,DickeySe2003}.

The constrained integrable hierarchy, what can be considered as the invariant submanifold of the KP hierarchy, encompasses many noteworthy systems such as the $GD_{n}$ (Gelfand-Dickey) hierarchy \cite{DickeySe2003,AdlerMoserOa1978,AdlerOa1979,YangCBmLMP2024}, CKP hierarchy \cite{DateKashiwaraJimboMiwaKh1981} and constrained KP hierarchy \cite{LiuBsJMP1996, AratynNPVsPLA1997, DickeyOtLMP1995,LiuZZCiJGP2015, WuOtJMP2012} and so on.
For example, $\{L\in\mathcal{N}|(L^{n})_{<0}=0\}$
is invariant under the KP flows, then the $GD_{n}$ hierarchy  \cite{DickeySe2003,AdlerMoserOa1978,AdlerOa1979} is
\begin{eqnarray}
\frac{\partial L}{\partial t_{j}}=[(L^j)_{\geq0},L],\;j\neq 0\; \text{mod} \;n,\nonumber
\end{eqnarray}
which can derive the KdV equation describing shallow water waves.
Let $^\ast$ denote the adjoint action with respect to the $L^2$ norm.
The set $\{L\in\mathcal{N}|L^{*}=-L\}$
is an invariant submanifold under the KP flows.
The $(2n+1)$-th flow of the CKP hierarchy \cite{DateKashiwaraJimboMiwaKh1981} is defined by
\begin{eqnarray}\label{CKPLax}
\frac{\partial L}{\partial t_{2i-1}}=[(L^{2i-1})_{\geq0},L],\;\;\;i=1,2,\cdots.
\end{eqnarray}
The Kaup-Kuperschmidt equation is its first non trivial equation that can be used to describe certain nonlinear waves in \cite{KaupOtSAM1980,KupershmidtAsPLA1984}.

Let
\begin{eqnarray}
\mathcal{N}_{cKP}=\{L\in\mathcal{N}|L^{k}=(L^{k})_{\geq0}+\sum_{i=1}^{N} \psi_{i} \partial^{-1}\phi_{i}\}.\nonumber
\end{eqnarray}
For $L\in\mathcal{N}_{cKP}$, if $\psi_i$'s and $\phi_i$'s are eigenfunctions and the adjoint eigenfunctions, submanifold $\mathcal{N}_{c}$ is an invariant under the KP flows. The hierarchy is called the constrained KP hierarchy\cite{KonopelchenkoSSDiPLA1991,ChengLTcPLA1991,AratynNPVsPLA1997, DickeyOtLMP1995,LiuZZCiJGP2015}.
The interplay between it and discrete multi-matrix models was described, and its additional symmetries and Darboux-B\"{a}cklund solutions were obtained \cite{AratynNPCKIJMPA1997}. Oevel and Strampp constructed bi-Hamiltonian structures and Wronskian solutions of the constrained KP hierarchy\cite{OevelSCkCMP1993,OevelSWsJMP1996,OevelDtPa1993}.
A geometric description of the constrained KP hierarchy was provided within the Grassmann manifold of Segal and Wilson, and its equivalence to Krichever's general rational reductions of the KP hierarchy was established \cite{HelminckvAaCMP1998}.

Let
\begin{eqnarray}
\mathcal{N}_{cCKP}=\{L\in\mathcal{N}|L^{2n+1}=(L^{2n+1})_{\geq0}+\sum_{i=1}^{m}\left(q_{i}\partial^{-1}r_{i}+r_{i}\partial ^{-1}q_{i}\right)\}.\nonumber
\end{eqnarray}
For $L\in\mathcal{N}_{cCKP}$, if the $q_{i}$'s and $r_{i}$'s are eigenfunctions, $\mathcal{N}_{cCKP}$ is an invariant submanifold under the KP flows.
The hierarchy is called the constrained CKP hierarchy \cite{LorisOr1999}. When $n=0$, gauge transformations and additional symmetries of constrained CKP hierarchy are constructed \cite{HeWCGtJMP2007,TianHCAsSCM2011}.

Soliton equations can be generated from a splitting of a Lie algebra \cite{TerngUhlenbeckTn2011, DrinfeldSokolovLa1984,LiuWZDhCMP2020}.
The $\hat{A}_{2n}^{(2)}$-KdV hierarchy was constructed by using splitting theory \cite{TerngWuDt2023} and it is equivalent to the $GD_{2n+1}$ hierarchy generated by
\begin{eqnarray}
\mathcal{A}=\left\{L=\partial^{2n+1}+\sum_{i=1}^{n}\left(\partial^{n+1-i}u_i\partial^{n-i}+\partial^{n-i}u_i\partial^{n+1-i}\right)\mid u_i \in \mathbb{C}(\mathbb{R}, \mathbb{C})\right\}.\nonumber
\end{eqnarray}
Darboux transformations for the $\hat{A}_{2n}^{(2)}$-KdV hierarchy are constructed and a geometric interpretation of the soliton hierarchy the equations was given \cite{TerngWuDt2023}.
In this paper, we will generalize the constrained CKP hierarchy and construct its matrix formulation. Subsequently, we investigate the Darboux transformations, permutability formula, scaling transformations and tau functions of the matrix constrained CKP hierarchy. In this way, we can study the Virasoro vector fields on $\ln \tau_f$ for the $\hat{A}_{2 n}^{(2)}$-KdV hierarchy.

The paper is organized as follows. In Section 2, we construct the matrix constrained CKP hierarchy by using the splitting theory and prove that the matrix constrained CKP hierarchy and the constrained CKP hierarchy are equivalent. The Darboux transformations, scaling transformations and tau functions are given in Section 3 and 4, respectively. In Section 5, we prove that the Virasoro vector fields on $\ln \tau_f$ for the $\hat{A}_{2 n}^{(2)}$-KdV hierarchy are given by partial differential operators. Section 6 is left for conclusions and dicsussions.

\bigskip

\section{\sc \bf Splitting Theory for the Matrix Constrained CKP Hierarchy}

In this section, we construct the matrix constrained CKP hierarchy by a splitting of the Lie algebra \cite{TerngUhlenbeckTn2011, DrinfeldSokolovLa1984,LiuWZDhCMP2020}.

We define an involution of $SL(2m+2n+1,\mathbb{C})$ as
\begin{eqnarray}
\sigma(Y)=\left(C_{m,n} Y^t C_{m,n}^{-1}\right)^{-1}, \;Y \in SL(2m+2 n+1, \mathbb{C}),\nonumber
\end{eqnarray}
where
\begin{eqnarray}
C_{m,n}=\operatorname{diag}(\alpha_m, \beta_n), \;\;\alpha_m=\sum_{i=1}^{2m}e_{i, 2m+1-i},\beta_n=\sum_{i=1}^{2n+1}(-1)^{n+i-1}e_{i, 2n+2-i}.\nonumber
\end{eqnarray}
Note that
\begin{eqnarray}
C_{m,n}^2=I_{2m+2n+1}.\nonumber
\end{eqnarray}
Then the induced involution $\sigma_*$ on $sl(2m+2 n+1, \mathbb{C})$ is
\begin{eqnarray}
\sigma_*(y)=-C_{m,n} y^t C_{m,n}^{-1}, \;y\in sl(2m+2 n+1, \mathbb{C}).\nonumber
\end{eqnarray}
Let $\mathcal{K}$ and $\mathcal{P}$ denote the 1 and $-1$ eigenspaces of $\sigma_*$ on $sl(2m+2n+1, \mathbb{C})$, respectively,
then
\begin{eqnarray}
[\mathcal{K}, \mathcal{K}] \subset \mathcal{K}, \quad[\mathcal{K}, \mathcal{P}] \subset \mathcal{P}, \quad[\mathcal{P}, \mathcal{P}] \subset \mathcal{K}.\nonumber
\end{eqnarray}
Let
\begin{eqnarray}
\mathcal{G}_{m,n}=\left\{\xi(\lambda)=\sum_{i \leq i_0} \xi_i \lambda^i \mid
\sigma_*(\xi(-\lambda))=\xi(\lambda),\;
 \xi_i \in sl(2m+2n+1, \mathbb{C})\right\},\nonumber
\end{eqnarray}
and
\begin{eqnarray}
 \left(\mathcal{G}_{m,n}\right)_{+}&=&\left\{\xi(\lambda)=\sum_{i \geq 0} \xi_i \lambda^i \in \mathcal{G}_{m,n}\right\},\nonumber \\
 \left(\mathcal{G}_{m,n}\right)_{-}&=&\left\{\xi(\lambda)=\sum_{i<0} \xi_i \lambda^i \in \mathcal{G}_{m,n}\right\}.\nonumber
\end{eqnarray}
The pair $\left(\left(\mathcal{G}_{m,n}\right)_+, \left(\mathcal{G}_{m,n}\right)_{-}\right)$ is a splitting of $\left(\mathcal{G}_{m,n}\right)$. Then $\xi(\lambda) \in \mathcal{G}_{m,n}$ if and only if $\xi_{2 i} \in \mathcal{K}$ and $\xi_{2 i+1} \in \mathcal{P}$ for all $i$.

Let $G_{m,n}$ and $(G_{m,n})_{ \pm}$ be formal Lie groups associated with the Lie algebras $\mathcal{G}_{m,n}$ and $(\mathcal{G}_{m,n})_{ \pm}$, respectively.
We take a vacuum sequence of the splitting as
\begin{eqnarray}\label{vacseq}
\mathcal{J}=\left\{J^{2i-1} \mid J=J(\lambda)=\operatorname{diag}(0_{2m},J_0(\lambda)),  i=1,2,\cdots\right\}
\end{eqnarray}
in $(G_{m,n})_{+}$,
where
\begin{eqnarray}\label{J0}
J_0(\lambda)=e_{1,2n+1}\lambda+b, \;\;b=\sum_{i=1}^{2n}e_{i+1,i}.
\end{eqnarray}
The phase space of the soliton flows is $C^{\infty}(\mathbb{C},\mathcal{M})$,
where
\begin{eqnarray}
\mathcal{M}=\left\{(gJg^{-1})_{+}\mid g\in (G_{m,n})_{-}\right\}.\nonumber
\end{eqnarray}
Let $Y_{m,n}=\left\{[J, g]_{+}\in sl(2m+2n+1)|g\in\left(\mathcal{G}_{m,n}\right)_{-}\right\}.$
By direct calculation, we can obtain that
\begin{eqnarray}
\mathcal{M}=J+Y_{m,n}.\nonumber
\end{eqnarray}

Given a map $q: \mathbb{R}\rightarrow Y_{m,n}$, there exists a unique $S(q,\lambda)$, such that
\begin{equation}\label{Lxq}
\left\{\begin{array}{l}
{\left[\partial_x+J+q, S(q,\lambda)\right]=0,} \\
S(q,\lambda) \text { is conjugate to } J  \text { by }  (G_{m,n})_{-}.
\end{array}\right.
\end{equation}
Note that
\begin{eqnarray}
S^{2n+2}(q, \lambda)=\lambda S(q,\lambda).\nonumber
\end{eqnarray}
The above solution $S(q,\lambda)$ can be rewritten as
\begin{eqnarray}
S(q, \lambda)=e_{2m+1,2m+2 n+1} \lambda+\sum_{i \leq 0} S_{1, i}(q) \lambda^i,\nonumber
\end{eqnarray}
where $S_{1, i}(q) '$s are differential polynomials in $q$ for all $i \leq 0$.
Furthermore,
\begin{eqnarray}
S^{2 j-1}(q, \lambda)=\sum_{i } S_{2 j-1, i}(q) \lambda^i.\nonumber
\end{eqnarray}
\begin{remark}
By \eqref{Lxq}, we have
\begin{eqnarray}
\left[\partial_x+J+q, S^{2 j-1}(q, \lambda)\right]=0.\nonumber
\end{eqnarray}
Then by comparing the coefficients of $\lambda^i$ in the above equation, we have
\begin{eqnarray}\label{eqbydu}
\left[\partial_x+\operatorname{diag}(0_{2m},b)+q, S_{2 j-1, i}(q)\right]=\left[S_{2 j-1, i-1}(q),e_{2m+1,2m+2 n+1}\right].
\end{eqnarray}
In particular,
\begin{eqnarray}
\left[\partial_x+\operatorname{diag}(0_{2m},b)+q, S_{2 j-1, 0}(q)\right]=\left[S_{2 j-1,-1}(q),e_{2m+1,2m+2 n+1}\right].\nonumber
\end{eqnarray}
Hence $\left[\partial_x+\operatorname{diag}(0_{2m},b)+q, S_{2j-1,0}(q)\right]$ lies in $Y_{m,n}$.
\end{remark}

\begin{definition}
For maps $q:\mathbb{R}^2\rightarrow Y_{m,n}$, the $(2j-1)-$th $\mathcal{G}_{m,n}$ flow is
\begin{eqnarray}\label{def1Y}
q_t=\left[\partial_x+\operatorname{diag}(0_{2m},b)+q, S_{2j-1,0}(q)\right].
\end{eqnarray}
\end{definition}

\begin{proposition}\label{propqud}

Let
\begin{eqnarray}
\mu_i & =&e_{i, 2m+2n+1}+(-1)^{n+1}e_{2m+1,m+i},\nonumber\\
\nu_i & =&(-1)^{n}e_{m+i, 2m+2n+1}-e_{2m+1, i},\nonumber\\
\omega_i & =&-e_{2m+n+1-i, 2m+n+i}-e_{2m+n+2-i, 2m+n+1+i},\nonumber\\
V_{m,n}&=&\left\{\sum_{i=1}^{m}q_i\mu_i+\sum_{i=1}^{m}r_{m+1-i}\nu_i+\sum_{i=1}^{n}u_i\omega_i\in sl(2m+2n+1)\right\},\nonumber\\
\mathcal{B}^{+}&=&\left\{\sum_{i=1}^{m}q_i\mu_i+\sum_{i=1}^{m}r_{i}\nu_i+\sum_{i,j=2m+1,i\leq j}^{2m+2n+1}c_{i,j}e_{i,j}\in sl(2m+2n+1)\right\},\nonumber\\
\mathcal{N}^{+}&=&\left\{\sum_{i=1}^{m}q_i\mu_i+\sum_{i=1}^{m}r_{i}\nu_i+\sum_{i,j=2m+1,i< j}^{2m+2n+1}c_{i,j}e_{i,j}\in sl(2m+2n+1)\right\}.\nonumber
\end{eqnarray}
Let $B^{+}$ and $N^{+}$ be Lie groups associated with the Lie algebras $\mathcal{B}^{+}$ and $\mathcal{N}^{+}$, respectively.
Let $N_{n}^{+}$ be the upper triangular matrix of order $n$ with 1 on the diagonal.
Given $q \in C^{\infty}\left(\mathbb{R},Y_{m,n}\right)$, there exist a unique $\Delta \in C^{\infty}\left(\mathbb{R}, N^{+}\right)$ and $u \in$ $C^{\infty}\left(\mathbb{R}, V_{m,n}\right)$ such that
\begin{eqnarray}
\Delta\left(\partial_x+\operatorname{diag}(0_{2m},b)+q\right) \Delta^{-1}=\partial_x+\operatorname{diag}(0_{2m},b)+u.\nonumber
\end{eqnarray}
Moreover,
\begin{eqnarray}
\Delta\left(\partial_x+J+q\right) \Delta^{-1}=\partial_x+J+u.\nonumber
\end{eqnarray}
We write $u$ as $\Delta\ast q$.
\end{proposition}

\begin{proof}
Set
\begin{gather*}
q=\left(\begin{array}{ll}
0 & q_{12} \\
q_{21} & q_{22}
\end{array}\right)\in Y_{m,n}
\end{gather*}
with $q_{12}\in \mathbb{C}^{2m \times(2n+1)}, q_{21} \in \mathbb{C}^{(2n+1) \times 2m}, q_{22} \in \mathbb{C}^{(2n+1) \times(2n+1)},$
and
\begin{gather*}
\Delta=\left(\begin{array}{ll}
I_{2m} & \Delta_{12} \\
\Delta_{21} & \Delta_{22}
\end{array}\right)\in N^{+}
\end{gather*}
with $\Delta_{22}\in N_{2n+1}^{+}$.
We can easily obtain that $\Delta_{21}=\Delta_{12}=0$.
It can be shown through direct computation that there exist a unique $\Delta_{22}\in N_{2n+1}^{+}$, such that
\begin{eqnarray}
\Delta_{22}\left(\partial_x+b+q_{22}\right) \Delta^{-1}_{22}=\partial_x+b+u_{22}.\nonumber
\end{eqnarray}
\end{proof}

For $u\in V_{m,n}$, we can compute the induced cross section flow on $\partial_x+J+u$.

\begin{thm}\label{thmuC}
For $u\in V_{m,n}$, there exists a unique matrix $\eta_{2j-1}(u)\in C^{\infty}\left(\mathbb{R},\mathcal{N}^{+}\right)$ such that
\begin{eqnarray}
\left[\partial_x+\operatorname{diag}(0_{2m},b)+u, S_{2j-1,0}(u)+\eta_{2j-1}(u)\right] \in C^{\infty}\left(\mathbb{R}, V_{m,n}\right).\nonumber
\end{eqnarray}

\end{thm}

\begin{proof}
since $\left[\partial_x+J+u, S(u,\lambda)\right]=0,$ we have
\begin{gather*}
\left[\partial_x+\operatorname{diag}(0_{2m},b)+u, S_{2j-1,0}(u)\right]=-[e_{2m+1,2m+2n+1},S_{2j-1,-1}(u)]\in Y_{m,n},
\end{gather*}
Let
\begin{gather*}
\left(\begin{array}{ll}
0 & H_{12} \\
H_{21} & H_{22}
\end{array}\right)=-[e_{2m+1,2m+2n+1},S_{2j-1,-1}(u)]
\end{gather*}
with $ H_{12}\in \mathbb{C}^{2m \times(2n+1)}, H_{21}\in \mathbb{C}^{(2n+1) \times 2m}, H_{22}\in \mathbb{C}^{(2n+1) \times(2n+1)}.$

Let
\begin{gather*}
u=\left(\begin{array}{ll}
0 & u_{12} \\
u_{21} & u_{22}
\end{array}\right)\in V_{m,n}
\end{gather*}
with $ u_{12}\in \mathbb{C}^{2m \times(2n+1)}, u_{21}\in \mathbb{C}^{(2n+1) \times 2m}, u_{22}\in \mathbb{C}^{(2n+1) \times(2n+1)},$
\begin{gather*}
\eta_{2j-1}=\left(\begin{array}{ll}
0 & \eta_{12} \\
\eta_{21} & \eta_{22}
\end{array}\right)\in C^{\infty}\left(\mathbb{R},\mathcal{N}^{+}\right)
\end{gather*}
with $ \eta_{12}\in \mathbb{C}^{2m \times(2n+1)}, \eta_{21}\in \mathbb{C}^{(2n+1) \times 2m}, \eta_{22}\in \mathbb{C}^{(2n+1) \times(2n+1)},$

Let
\begin{gather*}
\left(\begin{array}{ll}
0 & D_{12} \\
D_{21} & D_{22}
\end{array}\right)=[\partial_x+\operatorname{diag}(0_{2m},b)+u,\eta_{2j-1}(u)]-[e_{2m+1,2m+2n+1},S_{2j-1,-1}(u)],
\end{gather*}
we have
\begin{eqnarray}
D_{12}&=&\eta_{12, x}+u_{12}\eta_{22}-\eta_{12}(b+u_{22})+H_{12},\nonumber\\
D_{21}&=&\eta_{21,x}+(b+u_{22})\eta_{21}-\eta_{22}u_{21}+H_{21},\nonumber\\
D_{22}&=&\eta_{22,x}+u_{21}\eta_{12}+(b+u_{22})\eta_{22}-\eta_{21}u_{12}-\eta_{22}(b+u_{22})+H_{22}.\nonumber
\end{eqnarray}
From the first and second equations above, it follows that $\eta_{12}= 0$ and $\eta_{21} = 0$.
Let $\mathcal{G}_i=gl(2n+1)\bigcap \operatorname{\{e_{j,j+i}\}}$ and $\xi_{\mathcal{G}_i}$ denote the $\mathcal{G}_i$-component of $\xi\in gl(2n+1)$. Form the third equation above, we get
\begin{gather}\label{ggDH}
[b,\eta_{22}]_{\mathcal{G}_i}=(D_{22})_{\mathcal{G}_i}-(\eta_{22,x})_{\mathcal{G}_i}-[u_{22},\eta_{22}]_{\mathcal{G}_i}-(H_{22})_{\mathcal{G}_i}.
\end{gather}
For $i=0$, we can use $(H_ {22}) _ {\mathcal {G}_1 }$ to uniquely determine $(\eta_{22})_{\mathcal{G}_1}$. For $i=1$,
\begin{gather*}
[b,\eta_{22}]_{\mathcal{G}_1}=(D_{22})_{\mathcal{G}_1}-(\eta_{22,x})_{\mathcal{G}_1}-[u_{22},\eta_{22}]_{\mathcal{G}_1}-(H_{22})_{\mathcal{G}_1}.
\end{gather*}
Direct calculation shows that $(\eta_{22})_{\mathcal{G}_2}$ and $(D_{22})_{\mathcal{G}_1}$ are uniquely determined by $(H_ {22}) _ {\mathcal {G}_1 }$ and $(\eta_{22})_{\mathcal{G}_1}$. 
By induction for $i$ in \eqref{ggDH}, we prove that $\eta_{22}$ can be uniquely determined.
\end{proof}

\begin{definition}
For maps $u:\mathbb{R}^2\rightarrow V_{m,n}$, the $(2j-1)$-th $\mathcal{G}_{m,n}$-KdV flow is
\begin{eqnarray}\label{def2V}
u_t=\left[\partial_x+\operatorname{diag}(0_{2m},b)+u, S_{2j-1,0}(u)+\eta_{2j-1}(u)\right],
\end{eqnarray}
where $\eta_{2j-1}$ is given by Theorem \ref{thmuC}.
\end{definition}

Hence, we get Lax expression of the $(2j-1)$-th $\mathcal{G}_{m,n}$-flow \eqref{def1Y} and the $(2j-1)$-th $\mathcal{G}_{m,n}$-KdV flow \eqref{def2V}.

\begin{thm}\label{thmFq}
The following statements are equivalent for map $q: \mathbb{R}^2 \rightarrow Y_{m,n}$,
\begin{itemize}
\item[\rm(1).] $q$ is a solution of the $(2j-1)$-th $\mathcal{G}_{m,n}$-flow \eqref{def1Y},

\item[\rm(2).] $q_t=\left[\partial_x+J(\lambda)+q,\left(S^{2 j-1}(q, \lambda)\right)_{+}\right]$,

\item[\rm(3).]
$
\left[\partial_x+J(\lambda)+q, \partial_t+\left(S^{2 j-1}(q, \lambda)\right)_{+}\right]=0,
$

\item[\rm(4).] The following linear system
\begin{equation}
\left\{\begin{array}{l}
 F^{-1}F_x=J(\lambda)+q, \\
F^{-1}F_t=\left(S^{2 j-1}(q, \lambda)\right)_{+},
\end{array}\right.\label{thmFq4}
\end{equation}
is solvable for $F(x, t, \lambda)\in SL(2m+2n+1,\mathbb{C})$ satisfying equations
\begin{eqnarray}\label{Fxlam}
\sigma(F(x, t, -\lambda))=F(x, t, \lambda).
\end{eqnarray}
\end{itemize}
\end{thm}

\begin{thm}\label{thmEu}
The following statements are equivalent for map $u: \mathbb{R}^2 \rightarrow V_{m,n}$,

\begin{itemize}
\item[\rm(1).]  $u$ is a solution of the $(2j-1)$-th $\mathcal{G}_{m,n}$-KdV flow \eqref{def2V},

\item[\rm(2).]  $u_t=\left[\partial_x+J(\lambda)+u,\left(S^{2 j-1}(u, \lambda)\right)_{+}+\eta_{2j-1}(u)\right]$,

\item[\rm(3).]
$
\left[\partial_x+J(\lambda)+u, \partial_t+\left(\left(S^{2 j-1}(u, \lambda)\right)_{+}+\eta_{2j-1}(u)\right)\right]=0,
$

\item[\rm(4).]  The following linear system
\begin{equation}
\left\{\begin{array}{l}
E^{-1}E_x=J(\lambda)+u, \\
E^{-1}E_t=\left(S^{2 j-1}(u, \lambda)\right)_{+}+\eta_{2j-1}(u),
\end{array}\right.\label{thmEu4}
\end{equation}
is solvable for $E(x, t, \lambda)\in SL(2m+2n+1,\mathbb{C})$ satisfying equations
\begin{eqnarray}\label{Exlam}
\sigma(E(x, t, -\lambda))=E(x, t, \lambda).
\end{eqnarray}
\end{itemize}
\end{thm}

We can establish the connection between the frames $F(x, t, \lambda)$ and $E(x, t, \lambda)$ by using Proposition \ref{propqud} and Theorems \ref{thmFq} and \ref{thmEu}.

\begin{proposition}\label{EF}
 Let $F(x, t, \lambda)$ given as \eqref{thmFq4}. If $q$ is a solution of the $(2j-1)$-th $\mathcal{G}_{m,n}$-flow \eqref{def1Y}, $\Delta$ defined as Proposition \ref{propqud}, then $E(x, t, \lambda)= F(x, t, \lambda)\Delta^{-1}(x, t)$ is a frame of the solution $u=\Delta\ast q$ of the $(2j-1)$-th $\mathcal{G}_{m,n}$-KdV flow \eqref{def2V}. Conversely,  Let $E(x, t, \lambda)$ given as \eqref{thmEu4}. If $u$ is a solution of the $(2j-1)$-th $\mathcal{G}_{m,n}$-KdV flow \eqref{def2V} and $\Delta(x, t)$ satisfying $\Delta_t\Delta^{-1}=\eta_{2j-1}(u)$, where $\eta_{2j-1}$ is given by Theorem \ref{thmuC}, then $F(x, t, \lambda)=\Delta(x, t) E(x, t, \lambda)$ is a frame of the solution $q=\Delta^{-1}\ast u$ of the $(2j-1)$-th $\mathcal{G}_{m,n}$-flow \eqref{def1Y}.
\end{proposition}

\begin{ex}\label{exYV}

\begin{enumerate}

\item[\rm(1).] Set $m=0$, \eqref{def1Y} and \eqref{def2V} give the $\hat{A}_{2 n}^{(2)}$- and $\hat{A}_{2 n}^{(2)}$-KdV flows \cite{TerngWuDt2023}, respectively.

\item[\rm(2).]  Set $n=0, m=1$, \eqref{def1Y} and \eqref{def2V} give the $2\times2$ AKNS hierarchy \cite{KonopelchenkoSSDiPLA1991,ChengLTcPLA1991}.

\item[\rm(3).]Set $n=1, m=1$,
let

\begin{gather*}
q=\left(\begin{array}{ccccc}
0 & 0 & 0 &0 & q_1 \\
0 & 0 & 0 &0 & -r_1 \\
-r_1 & q_1 & a &b & 0  \\
0 & 0 & 0 &0 & b \\
0 & 0 & 0 &0 & -a
\end{array}\right)\in Y_{1,1}
\end{gather*}

\begin{gather*}
u=\left(\begin{array}{ccccc}
0 & 0 & 0 &0 & q_1 \\
0 & 0 & 0 &0 & -r_1 \\
-r_1 & q_1 & 0 &-u_1 & 0  \\
0 & 0 & 0 &0 & -u_1 \\
0 & 0 & 0 &0 & 0
\end{array}\right)\in V_{1,1}
\end{gather*}

the third flow equation of \eqref{def1Y} is
\begin{equation}\label{def1Yqrab}
\left\{\begin{array}{l}
a_{t_3}=-2q_1r_1,\\
b_{t_3}=2aq_1r_1-q_{1,x}r_1-q_1r_{1,x},\\
q_{1,t_3}=q_{1,x x x}-2\left(a_x+b+\frac{a^2}{2}\right)q_{1,x}-q_1\left(a_x+b+\frac{a^2}{2}\right)_{x}, \\
r_{1,t_3}=r_{1,x x x}-2\left(a_x+b+\frac{a^2}{2}\right)r_{1,x}-r_1\left(a_x+b+\frac{a^2}{2}\right)_{x}.
\end{array}\right.
\end{equation}
By gauge transformation $u_1=-\left(a_x+b+\frac{a^2}{2}\right)$, we obtain
the third flow equation of \eqref{def2V} is
\begin{equation}\label{def1Vqru}
\left\{\begin{array}{l}
u_{1,t_3}=3q_1r_{1,x}+3q_{1,x}r_1,\\
q_{1,t_3}=q_{1,x x x}+2u_1q_{1,x}+ q_1u_{1,x}, \\
r_{1,t_3}=r_{1,x x x}+2u_1r_{1,x}+r_1u_{1,x}.
\end{array}\right.
\end{equation}

\end{enumerate}

\end{ex}

Next we explain the equivalence between constrained CKP hierarchy and the $\mathcal{G}_{m,n}$-KdV flows \eqref{def2V}.

Let $\mathcal{D}$ be as in section 1.
Let $\mathcal{M}$ be the algebra of the formal power series $Y=\sum Y_i\lambda^i, Y_i\in \left(\mathbb{R}, \mathbb{C}^{2m+2n+1}\right)$.  Let $\mathcal{E}$ be the algebra on $\mathcal{M}$ equipped with the equivalence relation $\sim$ satisfying that for any $\beta=(\beta_1,\beta_2,\cdots,\beta_{2m+2n+1})^t,\xi=(\xi_1,\xi_2,\cdots,\xi_{2m+2n+1})^t\in\mathcal{M}$, if 
\begin{gather*} 
\beta_{2m+1}=\xi_{2m+1},\beta_{2m+2}=\xi_{2m+2},\cdots,\beta_{2m+2n+1}=\xi_{2m+2n+1}
\end{gather*}
holds, then $\beta\sim\xi$. For $X=\sum_{i\leq i_{0}}X_i\partial^i, X_i\in C^{\infty}(\mathbb{R}, \mathbb{C})$, $\mathcal{D}_{<0}$ and $\mathcal{D}_{\geq0}$ denote $\sum_{i\leq 0}X_i\partial^i$ and $\sum_{i\geq0}X_i\partial^i$, respectively. Similarly, $\mathcal{E}_{+}$ and $\mathcal{E}_{-}$ are denoted. Specifically, for $X=\sum_{i\leq i_{0}}X_i\partial^i, X_i\in C^{\infty}(\mathbb{R}, \mathbb{C})$, $\mathbb{C}((\lambda^{-1}))_{[0]}=X_0.$

Let $\mathcal{L}=\partial+J+u$.
Referring to \cite{DrinfeldSokolovLa1984}, we introduce a $\mathcal{D}$-module structure on the space $\mathcal{E}$ of smooth functions from $\mathbb{R}$ to $\mathbb{C}^{2m+2n+1}$ as follows:
For $\eta \in C^{\infty}\left(\mathbb{R}, \mathbb{C}^{2m+2n+1}\right)$ and $P=\sum P_i\partial ^i \in \mathcal{D},$ we define
\begin{eqnarray}
P\diamond \eta = \sum P_i \mathcal{L}^i \eta.\nonumber
\end{eqnarray}
Set $M$ be a Baker function of $\mathcal{L}$, i.e., $M \in G_{-}$ and $\mathcal{L}=M^{-1}\left(\partial+J\right) M$.
By \eqref{def2V}, we can get
\begin{gather*}
\frac{\partial \mathcal{L}}{\partial t}=\left[\mathcal{L}, \left(M^{-1} J^{2j-1} M\right)_{+}-\tilde{\eta}_j(u)\right],
\end{gather*}
i.e.,
\begin{gather*}
\left[\frac{\partial }{\partial t}+\left(M^{-1} J^{2j-1} M\right)_{+}-\tilde{\eta}_j(u), \mathcal{L}\right]=0.
\end{gather*}

Furthermore, $\left[\mathcal{L}, f\right]=\frac{\partial f}{\partial t},$ and
\begin{gather*}
\left[\frac{\partial }{\partial t}+\left(M^{-1} J^{2j-1} M\right)_{+}-\tilde{\eta}_j(u), f\right]=\frac{\partial f}{\partial t}
\end{gather*}
where $f\in \mathbb{C}^{\infty}\left(\mathbb{R}^2, \mathbb{C}\right)$. Then we introduce the structure of  a $\left(\partial, \partial_t\right)$-module on $\mathbb{C}^{\infty}\left(\mathbb{R}^2, \mathbb{C}^{2m+2n+1}\right)$ as follows:
For $e_{2m+1}\in C^{\infty}\left(\mathbb{R}, \mathbb{C}^{2m+2n+1}\right)$ and $P=\sum P_{i,j}\partial ^i \partial_t^j,$ we define
\begin{eqnarray}
P\diamond e_{2m+1}= \sum P_{i,l}\mathcal{L}^i \left(\frac{\partial}{\partial t}+\left(M^{-1} J^{2j-1} M\right)_{+}-\tilde{\eta}_j(u)\right)^le_{2m+1}\nonumber.
\end{eqnarray}

We can obtain the following two lemmas in \cite{DrinfeldSokolovLa1984}.

\begin{lemma}\label{Lamuniq}
Let $\eta=\left(\eta_1, \eta_2,\cdots, \eta_{2m+2n+1}\right)^t \in \mathcal{E}$. Then there exists a unique $P_{\mathcal{L}} \in \mathcal{D}$ such that
\begin{gather*}
\eta=P_{\mathcal{L}} \diamond e_{2m+1}.
\end{gather*}
\end{lemma}

\begin{lemma}\label{lem+-}
Set $g\in \mathcal{G}_{m,n}$ and $P_g\in \mathcal{D}$. If $P_g\diamond e_{2m+1}=ge_{2m+1}$, then
\begin{eqnarray}
\left(P_g\right)_{\geq0} \diamond e_{2m+1}=g_{+} e_{2m+1}, \nonumber
\end{eqnarray}
where $g_{+} \in \left(\mathcal{G}_{m,n}\right)_{ +}$ and $\left(P_g\right)_{ \geq0} \in \mathcal{D}_{\geq0}$.
\end{lemma}

Using the $\mathcal{D}$-module structure, we can obtain the following proposition.

\begin{proposition}\label{L_c}
\begin{eqnarray}
L_c=\partial^{2n+1}+\sum_{i=1}^{n}\left(\partial^{n+1-i}u_i\partial^{n-i}+\partial^{n-i}u_i\partial^{n+1-i}\right)+\sum_{i=1}^{m}\left(q_{i}\partial^{-1}r_{i}+r_{i}\partial ^{-1}q_{i}\right),\nonumber
\end{eqnarray}
then $L_c\diamond e_{2m+1}=\lambda e_{2m+1}$.
\end{proposition}

\begin{proof}
For $\mathcal{L}=\partial+J+u$, we have
\begin{eqnarray}
\mathcal{L}(e_{2m+i})&=&e_{2m+i+1}, i=1, 2, \cdots, n,\nonumber\\
\mathcal{L}(e_{2m+n+1})&=&e_{2m+n+2}-u_1e_{2m+n},\nonumber\\
\mathcal{L}(e_{2m+n+i})&=&e_{2m+n+i+1}-u_{i-1}e_{2m+n+3-i}-u_{i}e_{2m+n+1-i},i=2, 3, \cdots, n,\nonumber\\
\mathcal{L}(e_{2m+2n+1})&=&\lambda e_{2m+1}-u_ne_{2m+2}+\sum_{i=1}^{m}\left(q_ie_i+(-1)^nr_ie_{2m+1-i}\right).\nonumber
\end{eqnarray}
Then we can prove this proposition through the following properties as
\begin{eqnarray}
\mathcal{L}^{-1}(r_ie_{2m+1})=-e_{i},\mathcal{L}^{-1}(q_ie_{2m+1})=(-1)^ne_{2m+1-i}, i=1, 2, \cdots, m.\nonumber
\end{eqnarray}
\end{proof}

Below we present the relationship between constrained CKP hierarchy and the $\mathcal{G}_{m,n}$-KdV flows \eqref{def2V}.

\begin{thm}\label{thmeqv3.5}
Let
\begin{eqnarray}
u=\sum_{i=1}^{m}q_i\mu_i+\sum_{i=1}^{m}r_{m+1-i}\nu_i+\sum_{i=1}^{n}u_i\omega_i\nonumber
\end{eqnarray}
 be a solution of the the $\mathcal{G}_{m,n}$-KdV flow \eqref{def2V}, then
 \begin{eqnarray}
L_c=\partial^{2n+1}+\sum_{i=1}^{n}\left(\partial^{n+1-i}u_i\partial^{n-i}+\partial^{n-i}u_i\partial^{n+1-i}\right)+\sum_{i=1}^{m}\left(q_{i}\partial^{-1}r_{i}+r_{i}\partial ^{-1}q_{i}\right),\nonumber
\end{eqnarray}
is a solution of the constrained CKP hierarchy defined by
\begin{eqnarray}
\frac{\partial L_c}{\partial t_{2j-1}}=[(L_c^{\frac{2j-1}{2n+1}})_{\geq0},L_c],,\;\;j=1,2,\cdots,\nonumber
\end{eqnarray}
and vice versa. Therefore, the $\mathcal{G}_{m,n}$-KdV flows \eqref{def2V} defines is known as the matrix constrained CKP hierarchy.
\end{thm}

\begin{proof}
Let $\left(M^{-1} J M\right) e_{2m+1}=T \diamond e_{2m+1}$. Therefore, for any $k$,
\begin{gather*}
T^k \diamond e_{2m+1}=\left(M J^k M^{-1}\right) e_{2m+1}.
\end{gather*}
By \eqref{Lxq}, we can obtain that $M^{-1} J^{2n+1} Me_{2m+1}=\lambda e_{2m+1}$ on $\mathcal{M}$. From Proposition \ref{L_c}, we have $T=L_{c}^{\frac{1}{2n+1}}$.
Then
\begin{gather*}
\left(M^{-1} J^k M\right)_{+} e_{2m+1}=(L_{c}^{\frac{k}{2n+1}})_{\geq0} \diamond e_{2m+1}.
\end{gather*}
Moreover, in Proposition \ref{L_c}, we can calculate that
\begin{gather*}
\left(\frac{\partial}{\partial t}-(L_{c}^{\frac{2j-1}{2n+1}})_{\geq0}\right) \diamond e_{2m+1}=0.
\end{gather*}
Because of Lemma \ref{Lamuniq}, $\left[\frac{\partial}{\partial t}-(L_{c}^{\frac{2j-1}{2n+1}})_{\geq0}, L_c\right]=0$, i.e.,
\begin{gather*}
\frac{\partial L_c}{\partial t}=\left[(L_{c}^{\frac{2j-1}{2n+1}})_{\geq0}, L_c\right].
\end{gather*}
\end{proof}

\begin{ex}
Let $L=\partial^3+u_1\partial+\partial u_1 $. Since the Lax operator $L$ is invariant under \eqref{CKPLax}, the Lax equations can be defined as
\begin{eqnarray}
\frac{\partial L}{\partial t_{2n+1}}=[L^{\frac{2n+1}{3}},L],\;\;n=0,1,2,\cdots,\nonumber,
\end{eqnarray}
which is the $GD_3$ hierarchy. Its first nontrivial equation is the KK (Kupershmidt-Kaup) equation \cite{KaupOtSAM1980,KupershmidtAsPLA1984} as
\begin{eqnarray}
u_{1,t}=-\frac{1}{9}\left(u_{1,x x x x x}-10 u_1 u_{1,x x x}-25 u_{1,x} u_{1,x x}+2 u_1^2 u_{1,x}\right). \nonumber
\end{eqnarray}
\end{ex}

\begin{ex}
Let $L=\partial+q_1\partial^{-1}r_1+r_1\partial^{-1}q_1$. The third flow equation of constrained CKP hierarchy is
\begin{equation}
\left\{\begin{array}{l}
q_{1,t_3}=q_{1,x x x}+3q_1^2 r_{1,x}+9 q_1 r_1 q_{1,x}, \\
r_{1,t_3}=r_{1,x x x}+3r_1^2 q_{1,x}+9q_1r_1 r_{1,x}.
\end{array}\right.
\end{equation}
Furthermore, let $L=\partial^3+u_1\partial+\partial u_1+q_1\partial^{-1}r_1+r_1 \partial^{-1}q_1$. The third flow equation of constrained CKP hierarchy is
the same as \eqref{def1Vqru}.

\end{ex}

\section{\sc \bf Darboux Transformations and Scaling Transformations }

In this section, we will construct the Darboux transformations and scaling transformations of the matrix constrained CKP hierarchy.

Let $\hat{G}_{+}$ denote the group of holomorphic maps $f: \mathbb{C} \rightarrow S L(2m+2 n+1, \mathbb{C})$ satisfying conditions
\begin{eqnarray}\label{conditons}
\sigma(\xi(-\lambda))=\xi(\lambda),
\end{eqnarray}
and $\hat{G}_{-}$ the group of rational maps $f: \mathbb{C} \cup\{\infty\} \rightarrow SL(2m+2 n+1, \mathbb{C})$ satisfying \eqref{conditons} with $f(\infty)=\mathrm{I}$.
Let $\mathbb{C}^{m+n+1, m+n}$ denote the vector space $\mathbb{C}^{2m+2n+1}$ equipped with bilinear form
\begin{gather*}
\langle X, Y\rangle=X^tC_{m,n}Y.
\end{gather*}
The adjoint $A^{\sharp}$ of a linear operator $A$ on $\mathbb{C}^{m+n+1, m+n}$ is defined by
\begin{gather*}
\langle A X , Y\rangle=\left\langle X, A^{\sharp}Y\right\rangle
\end{gather*}
for all $X, Y \in \mathbb{C}^{m+n+1, m+n}$.

Next, we construct simple elements in $\hat{G}_{-}$.

\begin{definition}
Let $\mathbb{C}^{n+1, n}=V \oplus V^{\perp}$ denote an orthogonal decomposition. A linear map $\pi: \mathbb{C}^{n+1, n} \rightarrow V$ is called an $O(m+n+1, m+n)$ projection onto $V$ along $V^{\perp}$ if it satisfies $\pi^2=\pi=\pi^{\sharp}$.
\end{definition}

Note that, if $\pi$ is the $O(m+n+1, m+n)$ projection onto $V$ along $V^{\perp}$, then
\begin{gather*}
\pi ^{\perp}:= I_{2m+2n+1}-\pi
\end{gather*}
is the $O(m+n+1, m+n)$ projection onto $V^{\perp}$ along $V$.

By direct calculation, we can obtain the following proposition.
\begin{proposition}\label{prop_kk}
Let $\pi$ be an $O(m+n+1, m+n)$ projection of $\mathbb{C}^{n+1, n}$ onto $V$ along $V^{\perp}$, $\alpha \in \mathbb{C} \backslash\{0\}$ a constant, and
\begin{gather*}
k_{\alpha, \pi}(\lambda)=\mathrm{I}_{2m+2n+1}+\frac{2 \alpha}{\lambda-\alpha}(\mathrm{I}-\pi).
\end{gather*}
Then $k_{\alpha, \pi} \in \hat{G}_{-}$ and $k_{\alpha, \pi}^{-1}=k_{-\alpha, \pi}.$
\end{proposition}

we use the local factorization method to construct the Darboux transformations of the matrix constrained CKP hierarchy.

\begin{proposition}\label{gFFg}
 Let $\alpha \in \mathbb{C} \backslash\{0\}$  and $\pi, \tilde{\pi}$ be an $O(m+n+1, m+n)$ projection of $\mathbb{C}^{n+1, n}$. Let $f(\lambda)$ be a meromorphic map satisfying \eqref{conditons}. If $f(\lambda)$ is holomorphic at $\lambda=\alpha,-\alpha$ and $\tilde{\pi}$ is an $O(m+n+1, m+n)$ projection with $\operatorname{Im} \tilde{\pi}=f^{-1}(\alpha)(\operatorname{Im} \pi)$, then
$\tilde{f}(x, t, \lambda)=k_{\alpha, \pi} f(\lambda) k_{\alpha, \tilde{\pi}}^{-1}$ is holomorphic at $\lambda=\alpha,-\alpha$.
\end{proposition}

\begin{proof}
We must demonstrate that the residues of
\begin{gather*}
\tilde{f}=\left(\mathrm{I}_{2m+2 n+1}-\frac{2 \alpha}{\lambda+\alpha} \tilde{\pi}^{\perp}\right) f(\lambda)\left(\mathrm{I}_{2m+2 n+1}+\frac{2 \alpha}{\lambda-\alpha} \pi^{\perp}\right)
\end{gather*}
at $\lambda=\alpha,-\alpha$ are zero.

For $\lambda=\alpha$,
the residue of $\tilde{f}$ is $2 \alpha\pi^{\perp} f(\alpha) \tilde{\pi}$. Since $f(\alpha)\operatorname{Im} \tilde{\pi}=(\operatorname{Im} \pi)$, we have $f(\alpha)\tilde{V}=V$ and $\pi^{\perp} \tilde{V}=0$, the residue is zero.

 For $\lambda=-\alpha$, the residue of $\tilde{f}$ at $\lambda=-\alpha$ is $-2 \alpha\pi f(-\alpha)\tilde{\pi}^{\perp}$.
Due to
\begin{eqnarray}
\left\langle f(\alpha)V^{\perp},\tilde{V}\right\rangle &=& \left\langle V^{\perp}, C_{m,n}f(\alpha)C_{m,n}\tilde{V}\right\rangle\nonumber\\
&=&\left\langle V^{\perp}, f(-\alpha)^{-1}\tilde{V}\right\rangle\nonumber\\
&=&0,\nonumber
\end{eqnarray}
we can get $f(-\alpha)\left(V^{\perp}\right)=\tilde{V}^{\perp}$.
So the residue of $\tilde{f}$ at $\lambda=\alpha$ is also zero.
\end{proof}

Using Propositions \ref{gFFg}, the Darboux transformations are obtained by factoring the product of a simple element and the frame.

\begin{thm}\label{thmFqKpi}
Let $F(x, t, \lambda)$ be the frame of a solution $q$ of the $(2j-1)$-th $\mathcal{G}_{m,n}$-flow \eqref{def1Y}, $\pi$ an $O(m+n+1, m+n)$ projection, $\alpha \in \mathbb{C} \backslash\{0\}$ a constant, $k_{\alpha, \pi}$ defined as in Proposition \ref{gFFg}, and $\tilde{V}(x, t)=F(x, t, \alpha)^{-1}(\operatorname{Im} \pi)$. Assume that there exists an open neighborhood $\mathcal{O}$ of the origin in $\mathbb{C}^2$ such that the restriction of $\langle,\rangle$ to $\tilde{V}(x, t)$ is non-degenerate for all $(x, t) \in \mathcal{O}$. Let $\tilde{\pi}(x, t)$ denote the $O(m+n+1, m+n)$ projection onto $\tilde{V}(x, t)$. Then
\begin{gather*}
\tilde{q}=q+2 \alpha\left[e_{2m+1,2m+2 n+1},\tilde{\pi}\right]
\end{gather*}
is a solution of the $(2j-1)$-th $\mathcal{G}_{m,n}$-flow \eqref{def1Y} defined on $\mathcal{O}$ and
\begin{gather*}
\tilde{F}(x, t, \lambda)=k_{\alpha, \pi}(\lambda)F(x, t, \lambda)k_{\alpha, \tilde{\pi}}^{-1}(\lambda)
\end{gather*}
is a frame of $\tilde{q}$. Let $g \bullet q$ denote the solution $\tilde{q}$. If $g_1, g_2 \in \hat{G}_-$, then $\left(g_1 g_2\right) \bullet q=g_1 \bullet\left(g_2 \bullet q\right)$.
\end{thm}

\begin{thm}\label{thmeqpixt}
Let $q$ be a solution of the $(2j-1)$-th $\mathcal{G}_{m,n}$-KdV flow \eqref{def1Y}, $\alpha \in \mathbb{C} \backslash 0$ a constant, and $\pi$ the set of $O(m+n+1, m+n)$ projections of $\mathbb{C}^{m+n+1, m+n}$. Then $\tilde{\pi}$ is the solution of
\begin{eqnarray}
\left\{\begin{array}{l}
\tilde{\pi}_x=\left(J(\alpha)+q+2 \alpha\left[e_{2m+1,2m+2 n+1},\tilde{\pi}\right]\right)\tilde{\pi}^{\perp}-\tilde{\pi}^{\perp}\left(J(\alpha)+q\right), \\
\tilde{\pi}_t=\left(S^{2j-1}\left(q+2 \alpha\left[e_{2m+1,2m+2 n+1},\tilde{\pi}\right],\alpha\right)\right)_{+}\tilde{\pi}^{\perp}-\tilde{\pi}^{\perp}\left(S^{2j-1}\left(q,\alpha\right)\right)_{+}.
\end{array}\right.\nonumber
\end{eqnarray}
\end{thm}

\begin{proof}
In Theorem \ref{thmFqKpi}, we note that
\begin{gather*}
\tilde{F}(x, t, \lambda)=k_{\alpha, \pi}(\lambda)F(x, t, \lambda)k_{\alpha,\tilde{\pi}}^{-1}(\lambda) .
\end{gather*}
By \eqref{thmFq4}, then
\begin{gather*}
\left(J(\alpha)+\tilde{q}\right)k_{\alpha, \bar{\pi}}=\frac{2\alpha}{\lambda-\alpha}\tilde{\pi}_x+k_{\alpha, \bar{\pi}}\left(J(\alpha)+q\right).
\end{gather*}
By using the residue at $\lambda=\alpha$ in the equation above, we can obtain that the first equation holds.
In a similar vein, we can demonstrate that the second equation is true.
\end{proof}

We can derive the Darboux transformations of the matrix constrained CKP hierarchy \eqref{def2V} by using Proposition \ref{EF} and Theorem \ref{thmeqpixt}.

\begin{thm}\label{thmDT4.7}
Let $E(x, t, \lambda)$ be a frame of a solution $u$ of the $(2j-1)$-th $\mathcal{G}_{m,n}$-KdV flow \eqref{def2V}, and $k_{\alpha, \pi}$, $k_{\alpha, \tilde{\pi}}$, $F(x, t, \lambda)$, $\tilde{F}(x, t, \lambda)$, $q$, $\tilde{q}$ as in Theorem \ref{thmFqKpi}, $\Delta(x, t)$ satisfying $\Delta_t\Delta^{-1}=\eta_{2j-1}(u)$. $\tilde{\Delta}$ and $\tilde{u}$ are given by $\tilde{q}$ and Proposition \ref{propqud}. Then
\begin{gather*}
\tilde{E}(x, t, \lambda)=k_{\alpha, \pi}(\lambda)E(x, t, \lambda)\Delta k_{\alpha,\tilde{\pi}}^{-1}(\lambda)\tilde{\Delta}^{-1}
\end{gather*}
is a frame of the solution $\tilde{u}=\tilde{\Delta}\ast \left(k_{\alpha,\tilde{\pi}}\bullet (\Delta^{-1}\ast u)\right)$.
\end{thm}

We construct the solutions for the $3$-th $\mathcal{G}_{m,n}$-flow \eqref{def1Y} and obtain the solutions for the matrix constrained CKP hierarchy by using the preceding conclusions.

\begin{ex}

Set $m=n=1$ as in Example \ref{exYV}.
Note that
\begin{gather*}
F(x, t, \lambda)=\exp \left(J(\lambda) x+J(\lambda)^3 t\right)
\end{gather*}
is a frame of the solution $q=0$ of the third flow equation.

Using $\lambda=z^3$, $F(x, t, \lambda)$ can be rewritten as
\begin{eqnarray}
F(x, t, z^3)=H\operatorname{diag}\left(1, 1, \exp(z^3t+zx), \exp(z^3t+\varepsilon zx), \exp(z^3t+\varepsilon^2zx)\right)H^{-1},\nonumber
\end{eqnarray}
where $\varepsilon=e^{2 \pi i / 3} $ and
\begin{eqnarray}
H=\left(\begin{array}{ccccc}
1 & 0 & 0 & 0 & 0 \\
0 & 1 & 0 & 0 & 0 \\
0 & 0 & z^2 & \varepsilon^2z^2 & \varepsilon z^2 \\
0 & 0 & z & \varepsilon z& \varepsilon^2 z \\
0 & 0 & 1 & 1 & 1 \\
\end{array}\right).\nonumber
\end{eqnarray}

Let $\pi$ be the $O(3,2)$ projection onto $\mathbb{C} v$, where $v=(1,-1,0,0,1)^t$. By using $z=1$, we have
\begin{gather*}
\tilde{V}(x, t):=\mathbb{C} F(x, t, 1)^{-1}(v)=\mathbb{C} \xi(x, t),
\end{gather*}
where
$$
\begin{aligned}
\xi & =\left(1,-1, \xi_1, \xi_2, \xi_3\right)^t, \\
\xi_1 & =\frac{\mathrm{e}^{-t-x}}{3}+\frac{\left(-\sqrt{3} \sin \left(\frac{ \sqrt{3}x}{2}\right)-\cos \left(\frac{ \sqrt{3}x}{2}\right)\right) \mathrm{e}^{-t+\frac{x}{2}}}{3},\\
\xi_2 & =\frac{\mathrm{e}^{-t-x}}{3}+\frac{\left(\sqrt{3} \sin \left(\frac{ \sqrt{3}x}{2}\right)-\cos \left(\frac{ \sqrt{3}x}{2}\right)\right) \mathrm{e}^{-t+\frac{x}{2}}}{3} ,\\
\xi_3 & =\frac{\mathrm{e}^{-t-x}}{3}+\frac{2 \mathrm{e}^{-t+\frac{x}{2}} \cos \left(\frac{ \sqrt{3}x}{2}\right)}{3}.
\end{aligned}
$$

Then the $O(3,2)$ projection $\tilde{\pi}(x, t)$ onto $\tilde{V}(x, t)$ is
\begin{eqnarray}
\tilde{\pi}&=&\frac{1}{\langle\xi, \xi\rangle} \xi \xi^t C_{1,1},\nonumber\\
&=&\frac{1}{\langle\xi, \xi\rangle} \left(\begin{array}{ccccc}
-1 & 1 & -\xi_3 & \xi_2 & -\xi_1 \\
1 & -1 & \xi_3 & -\xi_2 & \xi_1 \\
-\xi_1 & \xi_1 & -\xi_1\xi_3 & \xi_1\xi_2 & -\xi_1\xi_1 \\
-\xi_2 & \xi_2 & -\xi_2\xi_3 & \xi_2\xi_2 & -\xi_2\xi_1 \\
-\xi_3 & \xi_3 & -\xi_3\xi_3 & \xi_3\xi_2 & -\xi_3\xi_1 \\
\end{array}\right),\nonumber
\end{eqnarray}
where $\langle\xi, \xi\rangle=-2-2\xi_1\xi_3+\xi_2^2$.
By Theorem \ref{thmFqKpi}, we get
\begin{eqnarray}
\tilde{q}=\frac{2}{\langle\xi, \xi\rangle} \left(\begin{array}{ccccc}
0 & 0 & 0 & 0 & -\xi_3 \\
0 & 0 & 0 & 0 & \xi_3 \\
\xi_3 & -\xi_3 & \xi_3\xi_3 & -\xi_3\xi_2 & 0 \\
0 & 0 & 0 & 0 & -\xi_2\xi_3 \\
0 & 0 & 0 & 0 & -\xi_3\xi_3 \\
\end{array}\right).\nonumber
\end{eqnarray}
Then the solution for the third flow equation \eqref{def1Yqrab} is
\begin{gather*}
q_1=r_1=\frac{-2\xi_3}{-2-2\xi_1\xi_3+\xi_2^2},\;
a=\frac{2\xi_3^2}{-2-2\xi_1\xi_3+\xi_2^2},\;
b=\frac{-2\xi_2\xi_3}{-2-2\xi_1\xi_3+\xi_2^2}.
\end{gather*}
Hence, for the matrix constrained CKP hierarchy, set $m=n=1$, the solution for the third flow equation \eqref{def1Vqru} is
\begin{gather*}
q_1=r_1=\frac{-2\xi_3}{-2-2\xi_1\xi_3+\xi_2^2  },\;u_1=\frac{12\xi_2\xi_3+12\xi_1\xi_2\xi_3^2-6\xi_2^3\xi_3-6\xi_3^4}{(-2-2\xi_1\xi_3+\xi_2^2)^2}.
\end{gather*}

Let $\pi$ be the $O(3,2)$ projection onto $\mathbb{C} \tilde{v}$, where $\tilde{v}=(1,2,0,0,1)^t$. By using $z=1$, we have
\begin{gather*}
\tilde{V}(x, t):=\mathbb{C} F(x, t, -1)^{-1}(\tilde{v})=\mathbb{C} \tilde{\xi}(x, t),
\end{gather*}
where
$$
\begin{aligned}
\tilde{\xi} & =\left(1,2, \xi_1, \xi_2, \xi_3\right)^t.
\end{aligned}
$$
Then the solution for the third flow equation \eqref{def1Yqrab} is
\begin{gather*}
q_1=\frac{-2\xi_3}{4-2\xi_1\xi_3+\xi_2^2},\;r_1=\frac{4\xi_3}{4-2\xi_1\xi_3+\xi_2^2},
a=\frac{2\xi_3^2}{4-2\xi_1\xi_3+\xi_2^2},\;
b=\frac{-2\xi_2\xi_3}{4-2\xi_1\xi_3+\xi_2^2}.
\end{gather*}
Hence, for the matrix constrained CKP hierarchy, set $m=n=1$, the solution for the third flow equation \eqref{def1Vqru} is
\begin{gather*}
q_1=\frac{-2\xi_3}{4-2\xi_1\xi_3+\xi_2^2},\;r_1=\frac{4\xi_3}{4-2\xi_1\xi_3+\xi_2^2},\;u_1=\frac{-24\xi_2\xi_3+12\xi_1\xi_2\xi_3^2-6\xi_2^3\xi_3-6\xi_3^4}{(4-2\xi_1\xi_3+\xi_2^2)^2}.
\end{gather*}

Let $\pi$ be the $O(3,2)$ projection onto $\mathbb{C} \tilde{v}$, where $\tilde{v}=(-\frac{\sqrt{2}}{2},\frac{\sqrt{2}}{2},1,1,1)^t$. By using $z=1$, we have
\begin{gather*}
\tilde{V}(x, t):=\mathbb{C} F(x, t, 1)^{-1}(\tilde{v})=\mathbb{C} \tilde{\xi}(x, t),
\end{gather*}
where
$$
\begin{aligned}
\tilde{\xi} & =\left(-\frac{\sqrt{2}}{2},\frac{\sqrt{2}}{2}, \exp{(-x-t)}, \exp{(-x-t)}, \exp{(-x-t)}\right)^t.
\end{aligned}
$$
Then the solution for the third flow equation \eqref{def1Yqrab} is
\begin{gather*}
q_1=r_1=\frac{\sqrt{2}\exp{(-x-t)}}{1+\exp{(-2x-2t)}},
a=\frac{2\exp{(-2x-2t)}}{1+\exp{(-2x-2t)}},b=\frac{-2\exp{(-2x-2t)}}{1+\exp{(-2x-2t)}}.
\end{gather*}
Hence, for the matrix constrained CKP hierarchy, set $m=n=1$, the solution for the third flow equation \eqref{def1Vqru} is
\begin{gather*}
q_1=r_1=\frac{\sqrt{2}}{2\cosh(x+t)},\;u_1=\frac{3}{1+\cosh(2x+2t)}.
\end{gather*}

The soliton solutions of the matrix constrained CKP hierarchy is presented below.
Let $\pi$ be the $O(3,2)$ projection onto $\mathbb{C} \tilde{v}$, where $\tilde{v}=(-\frac{1}{2},1,1,1,1)^t$. By using $z=1$, we have
\begin{gather*}
\tilde{V}(x, t):=\mathbb{C} F(x, t, -1)^{-1}(\tilde{v})=\mathbb{C} \tilde{\xi}(x, t),
\end{gather*}
where
$$
\begin{aligned}
\tilde{\xi} & =\left(-\frac{1}{2},1, \exp{(-x-t)}, \exp{(-x-t)}, \exp{(-x-t)}\right)^t.
\end{aligned}
$$
Then the solution for the third flow equation \eqref{def1Yqrab} is
\begin{gather*}
q_1=\frac{\exp{(-x-t)}}{1+\exp{(-2x-2t)}},\;r_1=\frac{2\exp{(-x-t)}}{1+\exp{(-2x-2t)}},
a=-b=\frac{2\exp{(-2x-2t)}}{1+\exp{(-2x-2t)}}.
\end{gather*}
Hence, for the matrix constrained CKP hierarchy, set $m=n=1$, the solution for the third flow equation \eqref{def1Vqru} is
\begin{gather*}
q_1=\frac{1}{2\cosh(x+t)},\;r_1=\frac{1}{\cosh(x+t)},\;u_1=\frac{3}{1+\cosh(2x+2t)}
\end{gather*}
\end{ex}

The above example shows that the condition $q_i=r_i$ is invariant under a certain Darboux transformation

\begin{proposition}
Let $\pi$ be the one-dimensional projection onto $\mathbb{C}v$ with $v=(v_1,v_2,\cdots,v_{m+n})^t$ and $E(x, t, \lambda)$ be a frame of a solution $u$ of the $(2j-1)$-th flow \eqref{def2V}. Let
\begin{gather*}
(\varsigma_1,\varsigma_2,\cdots,\varsigma_{m+n})^t=E(x, t, \alpha)^{-1}v.
\end{gather*}
If $(-1)^n\varsigma_i-\varsigma_{2m+1-i}=0$, then the Darboux transformation preserves the solution $q_i=r_i,i=1,2,\cdots,m$ of the matrix constrained CKP hierarchy.
\end{proposition}

\begin{proof}
Let $$(\varsigma_1,\varsigma_2,\cdots,\varsigma_{m+n})^t=E(x, t, \alpha)^{-1}v.$$ Then the $O(3,2)$ projection $\tilde{\pi}(x, t)$ onto $\tilde{V}(x, t)$ is
\begin{eqnarray}
\tilde{\pi}=\frac{1}{\langle \varsigma, \varsigma\rangle} \varsigma \varsigma^t C_{1,1},\nonumber
\end{eqnarray}
where $\varsigma=(\varsigma_1,\varsigma_2,\cdots,\varsigma_{m+n})^t$.
By calculation, we have
\begin{eqnarray}
\tilde{q}_i=q_i+(-1)^n\frac{2 \alpha \varsigma_i\varsigma_{2m+2n+1}}{\langle \varsigma, \varsigma\rangle},
\tilde{r}_i=r_i+\frac{2 \alpha \varsigma_{2m+1-i}\varsigma_{2m+2n+1}}{\langle \varsigma, \varsigma\rangle},i=1,2,\cdots,m.\nonumber
\end{eqnarray}
Invoking $(-1)^n\varsigma_i-\varsigma_{2m+1-i}=0$, we deduce the validity of the proposition.

\end{proof}

We study permutability formula for the $\mathcal{G}_{m,n}$-flows \eqref{def1Y}.
Let $\pi$ be an $O(m+n+1, m+n)$-projection onto $V_1$, and $\left\{v_1, \ldots, v_k\right\}$ a basis of $V_1$ such that $\left\langle v_i, v_j\right\rangle=\delta_{i j} \epsilon_i$, where $\epsilon_i= \pm 1$. Then
\begin{gather*}
\pi= \xi\epsilon \xi^t C_{m,n}, \quad \epsilon=\operatorname{diag}\left(\epsilon_1, \ldots, \epsilon_k\right), \quad \xi=\left(v_1, \ldots, v_k\right).
\end{gather*}

\begin{proposition}\label{prop_alpha_tau}
Let $\alpha_1, \alpha_2 \in \mathbb{C} \backslash\{0\}$ such that $\left|\alpha_1\right| \neq\left|\alpha_2\right|$, and $\pi_1, \pi_2$ be two $O(m+n+1, m+n)$ projections. Then

\begin{itemize}
\item[\rm(1).] $\phi\left(\alpha_1, \alpha_2, \pi_1, \pi_2\right):=\alpha_1-\alpha_2+2 \alpha_2 \pi_2-2 \alpha_1 \pi_1$ is invertible,

\item[\rm(2).]  $\tau_i:=\phi \pi_i \phi^{-1}$ is an $O(m+n+1, m+n)$-projection for $i=1,2$,

\item[\rm(3).]  $\operatorname{Im} \tau_1=k_{\alpha_2, \pi_2}\left(\alpha_1\right)\left(\operatorname{Im} \pi_1\right)$ and $\operatorname{Im} \tau_2=k_{\alpha_1, \pi_1}\left(\alpha_2\right)\left(\operatorname{Im} \pi_2\right)$,

\item[\rm(4).]  $k_{\alpha_2, \tau_2}k_{\alpha_1, \pi_1}=k_{\alpha_1, \tau_1}k_{\alpha_2, \pi_2}$.
\end{itemize}
\end{proposition}

\begin{proof}
It is simple to determine if \rm(1)  and \rm(2)  are accurate.

For \rm(3), let $V_i=\operatorname{Im} \pi_i$ for $i=1,2$. By \rm(2), we have $\operatorname{Im} \tau_i=\phi\left(V_i\right)$.
Since we can obtain
\begin{eqnarray}
\phi\left(V_1\right) & =&\left(\alpha_1-\alpha_2+2 \alpha_2 \pi_2-2 \alpha_1 \pi_1\right)\left(V_1\right) \nonumber\\
& =&-\left(\alpha_1-\alpha_2+2 \alpha_2 \pi_2\right)\left(V_1\right)\nonumber\\
& =&\left(-\alpha_1-\alpha_2+2 \alpha_2 \pi_2\right)\left(V_1\right), \nonumber
\end{eqnarray}
and
\begin{eqnarray}
k_{\alpha_2, \pi_2}\left(\alpha_1\right)\left(\operatorname{Im} \pi_1\right)=\left(-\alpha_1-\alpha_2+2 \alpha_2 \pi_2\right)\left(V_1\right),\nonumber
\end{eqnarray}
then $\operatorname{Im} \tau_1=k_{\alpha_2, \pi_2}\left(\alpha_1\right)\left(\operatorname{Im} \pi_1\right)$ holds. Similarly, we get $\operatorname{Im} \tau_2=k_{\alpha_1, \pi_1}\left(\alpha_2\right)\left(\operatorname{Im} \pi_2\right)$.

Using Proposition \ref{gFFg}, we may determine that \rm(3) and \rm(4) are equivalent.
\end{proof}

By Theorem \ref{thmFqKpi} and Propositions \ref{gFFg} and \ref{prop_alpha_tau}, we can prove the following theorem.

\begin{thm}\label{thmPF5.2}
Let $\alpha_i, \pi_i, \tau_i$ be as in Proposition \ref{prop_alpha_tau} for $i=1,2$, and $q$ be a solution of the $(2j-1)$-th $\mathcal{G}_{m,n}$-flow \eqref{def1Y}. Then

\begin{itemize}
\item[\rm(1).] $k_{\alpha_1, \tau_1}\bullet \left(k_{\alpha_2, \pi_2}\bullet q\right)=k_{\alpha_2, \tau_2}\bullet\left(k_{\alpha_1, \pi_1}\bullet q\right)$.

\item[\rm(2).]  Let $F(x, t, \lambda)$ be the frame of $q$ with $F(0,0, \lambda)=1_{2m+2 n+1}, \tilde{\pi}_i(x, t)$ the $O(m+n+1, m+n)$ projection onto $F\left(x, t, \alpha_i\right)^{-1}\left(\operatorname{Im} \pi_i\right)$ for $i=1,2$, and $\tilde{\tau}_i=\bar{\phi} \tilde{\pi}_i \bar{\phi}^{-1}$, where $\tilde{\phi}=\alpha_1-\alpha_2+2 \alpha_2 \tilde{\pi}_2-2 \alpha_1 \tilde{\pi}_1$. Then
$$
\begin{aligned}
q_i & :=k_{\alpha_i, \pi_i} \bullet q=q+2 \alpha_i\left[e_{2m+1,2m+2 n+1}, \tilde{\pi}\right], \quad i=1,2, \\
q_{12} & :=k_{\alpha_1, \tau_1} \bullet\left(k_{\alpha_2, \pi_2} \bullet q\right)=q_1+2 \alpha_2\left[e_{2m+1,2m+2 n+1}, \tilde{\tau}_2\right]=q_2+2 \alpha_1\left[e_{2m+1,2m+2 n+1}, \tilde{\tau}_1\right].
\end{aligned}
$$
\end{itemize}
\end{thm}

\begin{proposition}\label{propF_D}
Let $F(x, t, \lambda)$ be the frame of a solution $q$ of the $(2j-1)$-th $\mathcal{G}_{m,n}$-flow \eqref{def1Y} with $F(0,0, \lambda)=\mathrm{I}_{2m+2 n+1}$. Let $r \in \mathbb{C} \backslash\{0\}$ and
\begin{gather*}
D_{m,n}(r)=\operatorname{diag}\left(r^{n}I_{2m}, 1, r, \ldots, r^{2 n}\right).
\end{gather*}
Then $\hat{F}(x, t, \lambda):=(r \cdot F)(x, t, \lambda):=D_{m,n}^{-1}(r) F\left(r x, r^{2 j-1} t, r^{-(2 n+1)} \lambda\right) D_{m,n}(r)$ is the frame of a solution $(r \cdot q)(x, t):=r D_{m,n}^{-1}(r) q\left(r x, r^{2 j-1} t\right) D_{m,n}(r)$ of the $(2j-1)$-th $\mathcal{G}_{m,n}$-flow \eqref{def1Y} such that $\hat{F}(0,0, \lambda)=\mathrm{I}_{2m+2 n+1}$.
\end{proposition}

\begin{proof}
The two straightforward yet critical formulas are given as
\begin{eqnarray}\label{C_D}
D_{m,n}(r) C_{m,n} D_{m,n}(r)=r^{2 n} C_{m,n}
\end{eqnarray}
and
\begin{eqnarray}\label{C_D_J}
D_{m,n}^{-1}(r) J\left(r^{-(2 n+1)} \lambda\right) D_{m,n}(r)=r^{-1}J(\lambda).
\end{eqnarray}

Let
$$
\begin{aligned}
\tilde{q}(x, t) & =rD_{m,n}^{-1}(r) q\left(rx, r^{2 j-1} t\right) D_{m,n}(r), \\
\tilde{F}(x, t, \lambda) &=D_{m,n}^{-1}(r) F\left(rx, r^{2 j-1}t, r^{-(2 n+1)} \lambda\right) D_{m,n}(r).
\end{aligned}
$$
Since $F$ satisfy \eqref{Fxlam}, $\bar{F}$ also satisfy \eqref{Fxlam} by \eqref{C_D}.

According to Theorem \ref{thmFq}, we have
\begin{eqnarray}
\tilde{F}(x, t, \lambda)^{-1}\tilde{F}(x, t, \lambda)_x=rD_{m,n}^{-1}(r)\left(J(r^{-(2n+1)}\lambda)+q\left(rx, r^{2 j-1} t\right)\right)D_{m,n}(r).\nonumber
\end{eqnarray}
Using \eqref{C_D_J}, we get
\begin{eqnarray}
\tilde{F}(x, t, \lambda)^{-1}\tilde{F}(x, t, \lambda)_x=J(\lambda)+(r \cdot q)(x, t).\nonumber
\end{eqnarray}

Owing to Theorem \ref{thmFq},
\begin{eqnarray}
\tilde{F}(x, t, \lambda)^{-1}\tilde{F}(x, t, \lambda)_t=r^{2j-1}D_{m,n}^{-1}(r)\left(S^{2j-1}\left(r x, r^{2 j-1}t, r^{-(2 n+1)} \lambda\right)\right)D_{m,n}(r).\nonumber
\end{eqnarray}
Because of \eqref{Lxq}, we can obtain that
\begin{eqnarray}
rD_{m,n}^{-1}(r)\left(S\left(rx, r^{2 j-1}t, r^{-(2 n+1)} \lambda\right)\right)D_{m,n}(r)=S\left(\tilde{q}, \lambda\right).\nonumber
\end{eqnarray}
Then
\begin{eqnarray}
\tilde{F}(x, t, \lambda)^{-1}\tilde{F}(x, t, \lambda)_t=S^{2j-1}\left(\tilde{q}, \lambda\right).\nonumber
\end{eqnarray}
\end{proof}

Similarly, we give scaling transformations for the matrix constrained CKP hierarchy.

\begin{thm}\label{thmST6.2}
Let $E(x, t, \lambda)$ be the frame of a solution
\begin{eqnarray}
u&=&\sum_{i=1}^{m}q_i\mu_i+\sum_{i=1}^{m}r_{m+1-i}\nu_i+\sum_{i=1}^{n}u_i\omega_i\nonumber
\end{eqnarray}
of the $(2j-1)$-th $\mathcal{G}_{m,n}$-KdV flow \eqref{def2V} with $E(0,0, \lambda)=\mathrm{I}_{2m+2 n+1}$, where $\mu_i, \nu_i, \omega_i$ as in Proposition \ref{propqud}. Let $r \in \mathbb{C} \backslash\{0\}$ and
\begin{gather*}
D_{m,n}(r)=\operatorname{diag}\left(r^{n}I_{2m}, 1, r, \ldots, r^{2 n}\right).
\end{gather*}
Then $\hat{E}(x, t, \lambda):=(r \cdot E)(x, t, \lambda):=D_{m,n}^{-1}(r) E\left(r x, r^{2 j-1} t, r^{-(2 n+1)} \lambda\right) D_{m,n}(r)$ is the frame of a solution $(r \cdot u)(x, t):=r D_{m,n}^{-1}(r) u\left(r x, r^{2 j-1} t\right) D_{m,n}(r)$, i.e.,
\begin{eqnarray}
\left(r \cdot q_i\right)(x, t)&=&r^{n+1} q_i\left(rx, r^{2 j-1} t\right), \quad 1 \leq i \leq m, \nonumber\\
\left(r \cdot r_i\right)(x, t)&=&r^{n+1} r_i\left(rx, r^{2 j-1} t\right), \quad 1 \leq i \leq m, \nonumber\\
\left(r \cdot u_i\right)(x, t)&=&r^{2 i} u_i\left(rx, r^{2 j-1} t\right), \quad 1 \leq i \leq n, \nonumber
\end{eqnarray}
of the $(2j-1)$-th $\mathcal{G}_{m,n}$-KdV flow \eqref{def2V} such that $\hat{E}(0,0, \lambda)=\mathrm{I}_{2m+2 n+1}$.
\end{thm}

Next, we will present the relationship between scaling transformations and the Darboux transformations for the $(2j-1)$-th $\mathcal{G}_{m,n}$-flow \eqref{def1Y}.

\begin{proposition}
 Let $q$ be a solution of the $(2j-1)$-th $\mathcal{G}_{m,n}$-flow \eqref{def1Y}, $\pi$ the $O(m+n+1, m+n)$ projection onto $V$, and $k_{\alpha, \pi}$ as given in Proposition \ref{prop_kk}. Let $r \in \mathbb{C} \backslash\{0\}$, and
\begin{gather*}
D_{m,n}(r)=\operatorname{diag}\left(r^{n}I_{2m}, 1, r, \ldots, r^{2 n}\right).
\end{gather*}
Then

\begin{itemize}
\item[\rm(1).] $r \cdot \pi:=D_{m,n}(r)\pi D_{m,n}^{-1}(r)$ is an $O(m+n+1, m+n)$-projection onto $D_{m,n}(r)V$.

\item[\rm(2).]  $D_{m,n}(r)k_{1, \pi}(s \lambda) D_{m,n}^{-1}(r)=k_{s^{-1},r\cdot \pi},$

\item[\rm(3).] $r^{-1} \cdot\left(k_{1, \pi} \bullet(r \cdot q)\right)=k_{r^{-(2 n+1)}, r\cdot \pi} \bullet q$.
\end{itemize}

\end{proposition}

\begin{proof}
By \eqref{C_D}, we have
\begin{eqnarray}
 D_{m,n}(r)C_{m,n}=r^{2 n} C_{m,n}D_{m,n}(r)^{-1},\nonumber
\end{eqnarray}
then by direct calculation that we get $\left(r \cdot \pi\right)^2=r \cdot \pi,\;r \cdot \pi=C_{m,n}\left(r \cdot \pi\right)C_{m,n}$. (1) holds. Direct computation also confirms that (2) is true.

Due to Proposition \ref{propF_D},
\begin{gather*}
\bar{F}(x, t, \lambda)=D_{m,n}^{-1}(r) F\left(rx, r^{2 j-1} t, r^{-(2 n+1)} \lambda\right)D_{m,n}(r)
\end{gather*}
is the frame of a solution
\begin{gather*}
\bar{q}(x, t)=r \cdot q(x, t)=rD_{m,n}^{-1}(r) q\left(r x, r^{2 j-1} t\right)D_{m,n}(r)
\end{gather*}
and $\tilde{F}(0,0, \lambda)=I_{2m+2 n+1}$.

Let
\begin{gather*}
\tilde{V}(x, t)=\tilde{F}(x, t, 1)^{-1} V.
\end{gather*}
Because of Proposition \ref{gFFg},
\begin{gather*}
F_1(x, t, \lambda)=k_{1, \pi}(\lambda)\bar{F}(x, t, \lambda)k_{1, \tilde{\pi}(x, t)}^{-1}(\lambda)
\end{gather*}
is a frame for $q_1(x, t):=k_{1, \pi} \bullet \tilde{q}(x, t)$, where $\tilde{\pi}(x, t)$ is the $O(m+n+1, m+n)$ projection onto $\tilde{V}(x, t)$.

By Proposition \ref{propF_D}, we have
\begin{gather*}
F_2(x, t, \lambda)=D_{m,n}(r) F_1\left(rx, r^{2 j-1} t, r^{-(2 n+1)} \lambda\right) D_{m,n}^{-1}(r)
\end{gather*}
is a frame for $q_2(x, t):=r^{-1} \cdot q_1(x, t)$.
Then
{\small
\begin{eqnarray}
F_2(x, t, \lambda)&=&D_{m,n}(r)k_{1, \pi(r^{-1} x, r^{-(2 j-1)} t)}^{-1}\left(r^{2 n+1} \lambda\right) \tilde{F}\left(r^{-1} x, r^{-(2 j-1)} t, r^{2 n+1} \lambda\right) k_{1, \tilde{\pi}}(r^{2 n+1} \lambda)D_{m,n}^{-1}(r) \nonumber\\
&=&k_{r^{2 n+1}, r \cdot \pi} F(x, t, \lambda)k_{r^{-(2 n+1)}, r \cdot \tilde{\pi}\left(r x, r^{2 j-1} t\right)}^{-1}\nonumber.
\end{eqnarray}
}

Note that
\begin{eqnarray}
D_{m,n}(r) \bar{V}\left(r^{-1}x, r^{-(2 j-1)} t\right) &=&D_{m,n}(r) \bar{F}^{-1}\left(r^{-1}x, r^{-(2 j-1)} t, 1\right) V\nonumber \\
&=&F\left(x, t, r^{-(2 n+1)}\right)^{-1} \hat{V}.\nonumber
\end{eqnarray}
then $F_2(x, t, \lambda)$ is a frame for $k_{r^{-(2 n+1)}, r \cdot \pi} \bullet q(x, t)$. Therefore, (3) holds.

\end{proof}

\section{\sc \bf Tau Functions}

In this section, we investigate the relation between tau functions and solutions for the matrix constrained CKP hierarchy based on the algorithm given in \cite{TerngUKTf2016,TerngUKTf20162}.

We first review the construction of tau function. Let group $(G_{m,n})_{ \pm}$ and Lie algebra  $\left(\mathcal{G}_{m,n}\right)_{\pm}$ be as in section 2. The pair $\left(\left(\mathcal{G}_{m,n}\right)_{+}, \left(\mathcal{G}_{m,n}\right)_{-}\right)$ is a splitting of $\mathcal{G}_{m,n}$. Let $\left\langle, \right\rangle_{k}$ denote the bilinear form on $\mathcal{G}_{m,n}$ defined by
$$
\begin{aligned}
\langle\xi, \eta\rangle_k=\sum_j \operatorname{tr}\left(\xi_j \eta_{-j+k}\right),
\end{aligned}
$$
where $\xi(\lambda)=\sum_j \xi_j \lambda^j$ and $\eta(\lambda)=\sum_j \eta_j \lambda^j$.
Then $\left(\mathcal{G}_{m,n}\right)_{+}$ and $\left(\mathcal{G}_{m,n}\right)_{-}$  is compatible with the 2-cocycle $\omega$ defined by
$$
\begin{aligned}
w(\xi, \eta)=\left\langle\xi_\lambda, \eta\right\rangle_{-1},
\end{aligned}
$$
where $\xi, \eta \in \mathcal{G}_{m,n}$.
Let $\tilde{\mathcal{G}}$ is a central extension of $\mathcal{G}_{m,n}$ defined by $\omega$ and $\pi: \tilde{G} \rightarrow \mathcal{G}_{m,n}$ be  the projection. Then $\tilde{G}$ is a central extension of $G_{m,n}$. Thus, there is a natural lift $S: \tilde{G}_{+}\cup\tilde{G}_{-}\rightarrow \tilde{G}$ with $\pi \circ S=\operatorname{id}$.

Let the vacuum sequence $J^{2i-1}$ be as in \eqref{vacseq} and $W(t)=\exp(\sum_{i=1}^N t_{2i-1}J^{2i-1})$. Given $f\in \tilde{G}_{-}$, we can factor
\begin{gather*}
W(t) f^{-1}=M(t)^{-1} E(t), \quad \text { with } M(t) \in \tilde{G}_{-} \quad \text { and } \quad E(t) \in \tilde{G}_{+}.
\end{gather*}
Then tau function $\tau_f(t)$ is defined by
\begin{gather*}
S\left(W(t) \right) S\left(f^{-1}\right)=\tau_f(t) S\left(M(t)^{-1}\right) S\left(E(t)\right).
\end{gather*}
Based on the conclusions of Terng and Uhlenbeck \cite{TerngUKTf2016,TerngUKTf20162}, we can derive the formula for the tau function in the hierarchy defined by Definition \ref{def1Y} as follows
\begin{eqnarray}
& \left(\ln \tau_f\right)_{t_1 t_j}=\left\langle S^j(q,\lambda),e_{2m+1,2m+2n+1}\right\rangle_{-1},\nonumber \\
& \left(\ln \tau_f\right)_{t_j t_k}=\left\langle S^j(q,\lambda), \partial_\lambda\left(S^k(q,\lambda)\right)_{+}\right\rangle_{-1}.\label{tau1}
\end{eqnarray}
where $S(q,\lambda)$ was given by \eqref{Lxq}.

Next, we present the formula for the tau function $\tau_f$ of the matrix constrained CKP hierarchy.

\begin{thm}\label{thmtau11}
Let $\left(\left(\mathcal{G}_{m,n}\right)_{+}, \left(\mathcal{G}_{m,n}\right)_{-}\right)$ is a splitting of $\mathcal{G}_{m,n}$. Let the vacuum sequence $J^{2i-1}$ be as in \eqref{vacseq} and $f\in \tilde{G}_{-}$. Let  $u$ is a solution of the matrix constrained CKP hierarchy \eqref{def2V}. Then
$$
\begin{aligned}
& \left(\ln \tau_f\right)_{t_1 t_{2j-1}}=\left\langle S^{2j-1}(u,\lambda),e_{2m+1,2m+2n+1}\right\rangle_{-1}, \nonumber\\
& \left(\ln \tau_f\right)_{t_{2j-1} t_{2k-1}}=\left\langle S^{2j-1}(u,\lambda), \partial_\lambda\left(S^{2k-1}(u,\lambda)\right)_{+}\right\rangle_{-1},\\
\end{aligned}
$$
where $S(u,\lambda)$ is the solution of \eqref{Lxq} for $u\in V_{m,n}$.
\end{thm}

\begin{proof}
By Proposition \ref{propqud}, given $q \in C^{\infty}\left(\mathbb{R},Y_{m,n}\right)$, there exist a unique $\Delta \in C^{\infty}\left(\mathbb{R}, N^{+}\right)$ and $u \in$ $C^{\infty}\left(\mathbb{R}, V_{m,n}\right)$ such that
\begin{eqnarray}
\Delta\left(\partial_x+J+q\right) \Delta^{-1}=\partial_x+J+u.\nonumber
\end{eqnarray}
Then, in \eqref{Lxq}, given a map $u: \mathbb{R}\rightarrow V_{m,n}$, there exists a unique $\Delta S(q,\lambda) \Delta^{-1}$, such that
\begin{equation}
\left\{\begin{array}{l}
{\left[\partial_x+J+u, \Delta S(q,\lambda) \Delta^{-1}\right]=0,} \\
\Delta S(q,\lambda) \Delta^{-1} \text { is conjugate to } J  \text { by }  (G_{m,n})_{-},\nonumber
\end{array}\right.
\end{equation}
where $(G_{m,n})_{-}$ is given in section 2.
By the uniqueness of solutions, it follows that
$$S(u,\lambda) =\Delta S(q,\lambda) \Delta^{-1}.$$
Since $\Delta$ is independent of $\lambda$, we have
$$\Delta S_\lambda(q,\lambda) \Delta^{-1}=\left(\Delta S(q,\lambda) \Delta^{-1}\right)_\lambda.$$
Using \eqref{tau1}, we derive
\begin{gather*}
\left(\ln \tau_f\right)_{t_{2j-1} t_{2k-1}}=\left\langle S^{2j-1}(u,\lambda), \partial_\lambda\left(S^{2k-1}(u,\lambda)\right)_{+}\right\rangle_{-1}.
\end{gather*}

\end{proof}

We need to study $S(u,\lambda)$. We first present a fundamental lemma.

\begin{lemma}\label{MN2m}
Let $Mat(2m,2n+1,c)$ denote an $2m\times(2n+1)$ matrix where only the last column is non-zero and $Mat(2n+1,2m,r)$ denote an $(2n+1)\times 2m$ matrix where only the first row is non-zero. Let $J_0(\lambda)$ be as in \eqref{J0}. Any $M_{1,2}\in\{\sum_{i=0}^{+\infty}M_i\lambda^{-i}|M_0\in Mat(2m,2n+1,c),M_i\in\mathbb{C}^{2m\times(2n+1)}\}$ can be uniquely rewritten in the form
\begin{gather*}
M_{1,2}=\sum_{i=0}^{+\infty}m_iJ_0(\lambda)^{-i},
\end{gather*}
where $m_i\in Mat(2m,2n+1,c)$. Similarly, any $N_{2,1}\in\{\sum_{i=0}^{+\infty}N_i\lambda^{-i}|N_0\in Mat(2n+1,2m,r),N_i\in\mathbb{C}^{(2n+1)\times 2m}\}$ can be uniquely rewritten in the form
\begin{gather*}
N_{2,1}=\sum_{i=0}^{+\infty}J_0(\lambda)^{-i}n_i,
\end{gather*}
where $n_i\in Mat(2n+1,2m,r)$.
\end{lemma}

\begin{thm}\label{thmSSS}
Let $\varrho$ defined by $J_0(\lambda) h=h^\varrho J_0(\lambda)$, where $h^\varrho=\operatorname{diag}(h_2,h_3,\cdots,h_{2n+1},h_1)$ if $h^\varrho=\operatorname{diag}(h_1,h_2,\cdots,h_{2n+1})$. Let $S(u,\lambda)=\left(\begin{array}{ll}
S_{1,1} & S_{1,2} \\
S_{2,1} & S_{2,2}
\end{array}\right)\in \mathcal{G}_{m,n}$
 is the solution of \eqref{Lxq} for
\begin{gather*}
u=\sum_{i=1}^{m}q_i\mu_i+\sum_{i=1}^{m}r_{m+1-i}\nu_i+\sum_{i=1}^{n}u_i\omega_i\in V_{m,n}.
\end{gather*}
Then

\begin{itemize}
\item[\rm(1).]$S_{1,1}=\sum_{i=1}^{+\infty}S^{11}_{i}\lambda^{-i},$
where $S^{11}_{i}\in \mathbb{C}^{2m\times 2m}$, and its elements are polynomials in $q_1,\cdots$,$q_m$ and $r_1,\cdots,r_m$, with derivatives up to order $(2n+1)(i-1)$,

\item[\rm(2).]$S_{1,2}=\sum_{i=0}^{+\infty}m_iJ_0(\lambda)^{-i},$
where $m_i\in Mat(2m,2n+1,c)$, and its elements are polynomials in $q_1,\cdots,q_m$ and $r_1,\cdots,r_m$, with derivatives up to order $i$,

\item[\rm(3).]  $S_{2,1}=\sum_{i=0}^{+\infty}J_0(\lambda)^{-i}n_i,$
where $n_i\in Mat(2n+1,2m,r)$, and its elements are polynomials in $q_1,\cdots,q_m$ and $r_1,\cdots,r_m$, with derivatives up to order $i$,

\item[\rm(4).]  $S_{2,2}=J_0(\lambda)+\sum_{i=0}^{+\infty}h_i J_0(\lambda)^{-i},$
where $h_i$'s are diagonal matrices. Moreover, $h_0=0$,
$$
\begin{aligned}
&h_{2j-1}=\frac{u_j}{2n+1}\operatorname{diag}(-2n+1,-2n+1,2,\cdots,2)^{\varrho^{n-j}}+\varpi_{2j-1},\;\;j=1,2,\cdots,n,\\
&h_{2j}=\varpi_{2j},\;\;j=1,2,\cdots,n,
\end{aligned}
$$
where $\varpi_{2i}$ and $\varpi_{2i+1}$ are polynomials of $u_1,\cdots,u_{i-1}$ and their $x$ derivatives, and
elements of $h_{i},(i=2n+1,2n+2,\cdots)$ are polynomials in $q_1,\cdots,q_m$ and $r_1,\cdots,r_m$, with derivatives up to order $i-2n-1$.
\end{itemize}
\end{thm}

\begin{proof}
From Lemma \ref{MN2m}, we derive
\begin{gather*}
S_{1,2}=\sum_{i=0}^{+\infty}m_iJ_0(\lambda)^{-i},
S_{2,1}=\sum_{i=0}^{+\infty}J_0(\lambda)^{-i}n_i,
\end{gather*}
where $m_i\in Mat(2m,2n+1,c)$, $n_i\in Mat(2n+1, 2m,r)$. A straightforward computation demonstrates that
\begin{gather*}
S_{2,2}=J_0(\lambda)+\sum_{i=0}^{+\infty}h_i J_0(\lambda)^{-i},
\end{gather*}
where $h_i$'s are diagonal matrices.

Set
\begin{gather*}
u=\left(\begin{array}{ll}
0 & u_{12} \\
u_{21} & u_{22}
\end{array}\right)\in V_{m,n}
\end{gather*}
with $ u_{12}\in \mathbb{C}^{2m \times(2n+1)}, u_{21}\in \mathbb{C}^{(2n+1) \times 2m}, u_{22}\in \mathbb{C}^{(2n+1) \times(2n+1)}$.
By simple calculation we obtain that
\begin{gather*}
u_{22}=\sum_{i=1}^{n}A_{2i-1}J_0(\lambda)^{-(2i-1)},
\end{gather*}
where $A_{2i-1}=\operatorname{diag}(-u_i,-u_i,0,\cdots,0)^{\varrho^{n-1-i}}$.

Since $S(u,\lambda)$ is the solution of \eqref{Lxq} for $u\in V_{m,n}$, we have
\begin{equation}\label{su12}
\left\{\begin{array}{l}
S_{1,2, x}+u_{12}S_{2,2}-S_{1,1}u_{12}-S_{1,2}(J_0(\lambda)+u_{22})=0,\\
S_{2,2,x}+u_{21}S_{1,2}+(J_0(\lambda)+u_{22})S_{2,2}-S_{2,1}u_{12}-S_{2,2}(J_0(\lambda)+u_{22})=0,\\
\sum_{i_1,i_2,\cdots,i_{2n+1}=1\; \text{or} \;2}S_{1,i_1}S_{i_1,i_2}S_{i_2,i_3}\cdots S_{i_{2n+1},1}=\lambda S_{1,1},\\
\sum_{j_1,j_2,\cdots,j_{2n+1}=1 \;\text{or}\; 2}S_{2,j_1}S_{j_1,j_2}S_{j_2,j_3}\cdots S_{j_{2n+1},2}=\lambda S_{2,2}.
\end{array}\right.
\end{equation}

Comparing the coefficients of $J_0(\lambda)$ in the second equation and $J_0(\lambda)^{2n+1}$ in the fourth equation of \eqref{su12}, we deduce $h_0=0$.
For the case of $h_j(j=1,2,\cdots,2n)$, we prove it by induction. For $h_1$, the coefficients of $J_0(\lambda)^{0}$ and $J_0(\lambda)^{2n}$ in the second and fourth equations of \eqref{su12}, respectively, are compared to get
\begin{gather*}
h_1^{\varrho}-h_1+A_1-A_1^{\varrho}=0, h_{1}+h_{1}^{\varrho}+h_{1}^{\varrho^{2}}+\cdots+h_{1}^{\varrho^{2n}}=0.
\end{gather*}
So we have $h_{1}=\frac{u_1}{2n+1}\operatorname{diag}(-2n+1,-2n+1,2,\cdots,2)^{\varrho^{n-1}}$.
For $h_2$, the coefficients of $J_0(\lambda)^{-1}$ and $J_0(\lambda)^{2n-1}$ are compared to get
\begin{gather*}
h_2^{\varrho}-h_2+h_{1,x}=0,\; h_{1}+h_{2}^{\varrho}+h_{2}^{\varrho^{2}}+\cdots+h_{2}^{\varrho^{2n}}=0.
\end{gather*}
So we have $h_{2}=\varpi_{2}(u_1)$.
For $h_{2j-1}$, the coefficient of $J_0(\lambda)^{2-2j}$ in the second equation gives
\begin{gather*}
h_{2j-1}^{\varrho}-h_{2j-1}+A_j-A_j^{\varrho}+\Phi_{2j-1}(u_1,u_2,\cdots,u_{j-1})=0,
\end{gather*}
and the coefficient of $J_0(\lambda)^{2n+2-2j}$ in the second equation gives
\begin{gather*}
h_{2j-1}+h_{2j-1}^{\varrho}+h_{2j-1}^{\varrho^{2}}+\cdots+h_{2j-1}^{\varrho^{2n}}+\Psi_{2j-1}(u_1,u_2,\cdots,u_{j-1})=0.
\end{gather*}
The above two equations imply that
\begin{gather*}
h_{2j-1}=\frac{u_j}{2n+1}\operatorname{diag}(-2n+1,-2n+1,2,\cdots,2)^{\varrho^{n-i}}+\varpi_{2j-1}.
\end{gather*}
For $h_{2j-1}$, equating the coefficients of $J_0(\lambda)^{1-2j}$ and $J_0(\lambda)^{2n+1-2j}$ from the second and fourth equations in \eqref{su12}, respectively, we deduce that
\begin{gather*}
h_2^{\varrho}-h_2+\Phi_{2j}(u_1,u_2,\cdots,u_{j-1})=0,\; h_{1}+h_{2}^{\varrho}+h_{2}^{\varrho^{2}}+\cdots+h_{2}^{\varrho^{2n}}+\Psi_{2j}(u_1,u_2,\cdots,u_{j-1})=0.
\end{gather*}
These two equations show that $h_{2j}=\varpi_{2j}$.

From the coefficients of $J_0(\lambda)^{1}$ and $J_0(\lambda)^{0}$ of first equation of \eqref{su12}, it follows that
\begin{gather*}
m_0=u_{12},m_1=u_{12,x}.
\end{gather*}
Then by induction, we conclude that each elements of $m_i(i=1,2,3,\cdots,2n+1)$ are polynomials in $q_1,q_2,\cdots,q_m$ and $r_1,r_2,\cdots,r_m$, with derivatives up to order $i$.
By $S(u,\lambda)\in \mathcal{G}_{m,n}$, we find that $n_0=u_{21}$ and each elements of $n_i(i=1,2,3,\cdots,2n+1)$ are polynomials in $q_1,q_2,\cdots,q_m$ and $r_1,r_2,\cdots,r_m$, with derivatives up to order $i$.

By Lemma \ref{MN2m}, the left-hand side of the third equation of \eqref{su12} can be expressed as
\begin{gather*}
\sum_{i_1,i_2,\cdots,i_{2n+1}=1\; \text{or} \;2}S_{1,i_1}S_{i_1,i_2}S_{i_2,i_3}\cdots S_{i_{2n+1},1}=\sum_{j=-\infty}^{2n}m^{12}_{j}J_0(\lambda)^{j}n^{21}_{j}.
\end{gather*}
The $J_0(\lambda)^{2n}$ terms on the left-hand side lead to $S^{11}_{1}=u_{12}J_0(\lambda)^{2n}u_{21}$.

Following a similar approach to the above, we can complete the proof of this theorem.
\end{proof}

In the following theorem we give a relation between $\ln \tau_f$ and $u\in V_{m,n}$.

\begin{thm}\label{thmtaufunction2}
Let
\begin{gather*}
u=\sum_{i=1}^{m}q_i\mu_i+\sum_{i=1}^{m}r_{m+1-i}\nu_i+\sum_{i=1}^{n}u_i\omega_i
\end{gather*}
 is a solution of the matrix constrained CKP hierarchy \eqref{def2V}.
Then
$$
\begin{aligned}
&\left(\ln \tau_f\right)_{t_1 t_{2j-1}}=\frac{2(2j-1)}{2n+1}u_j+R_j,\;j=1,2,\cdots,n,\\
&\left(\ln \tau_f\right)_{t_1 t_{2n+2j-1}}=K_j,\;j=1,2,\cdots,
\end{aligned}
$$
where $R_j$'s are polynomials in $u_1,u_2,\cdots,u_{j-1}$ and their $x$ derivatives and $K_j$'s are polynomials in $q_1,q_2,\cdots,q_m$ and $r_1,r_2,\cdots,r_m$, with derivatives up to order $2(j-1)$.
\end{thm}

\begin{proof}
From Theorem \ref{thmtau11}, it follows that
\begin{gather*}
\left(\ln \tau_f\right)_{t_1 t_{2j-1}}=\left\langle S^{2j-1}(u,\lambda),e_{2m+1,2m+2n+1}\right\rangle_{-1},
\end{gather*}
Let $S^{2j-1}(u,\lambda)$ be denoted by
\begin{gather*}
\left(\begin{array}{ll}
S_{2j-1,11} & S_{2j-1,12} \\
S_{2j-1,21} & S_{2j-1,22}
\end{array}\right)
\end{gather*}
with $S_{2j-1,22}=\sum_{j}y_jJ_0(\lambda)^{j}\in \mathbb{C}^{(2n+1) \times(2n+1)}$,where $y_j$'s are diagonal matrices,
then we can derive
\begin{gather*}
\left(\ln \tau_f\right)_{t_1 t_{2j-1}}=\operatorname{tr}(y_{-1}e_{2n+1,2n+1}).
\end{gather*}
By direct calculation, we obtain
\begin{gather*}
S_{2j-1,22}=\sum_{i_1,i_2,\cdots,i_{2j-2}=1\; \text{or} \;2}S_{2,i_1}S_{i_1,i_2}S_{i_2,i_3}\cdots S_{i_{2j-2},2}.
\end{gather*}
If $j=1,2,\cdots,n$, the coefficients of $J_0(\lambda)^{-1}$ in $S_{2j-1,22}$ and $S_{2,2}^{2j-1}$ are identical and equal to
$$
\begin{aligned}
y_{-1}&=\sum_{i=0}^{2j-2}J_0(\lambda)^{2j-2-i}h_{2j-1}J_0(\lambda)^{-2j+i}+H_j(h_1,h_2,\cdots,h_{2j-2}),\\
&=\sum_{i=0}^{2j-2}h_{2j-1}^{\varrho^{2j-2-i}}+H_j(h_1,h_2,\cdots,h_{2j-2}),
\end{aligned}
$$
where $h_1,h_2,\cdots,h_{2j-1}$ are given by Theorem \ref{thmSSS}.
It follows that
\begin{gather*}
\left(\ln \tau_f\right)_{t_1 t_{2j-1}}=\frac{2(2j-1)}{2n+1}u_j+R_j,
\end{gather*}
where $R_j$'s are polynomials in $u_1,u_2,\cdots,u_{j-1}$ and their $x$ derivatives.

For the case of $2n+2j-1(j=1,2,\cdots)$, it follow from Theorem \ref{thmtau11} that we obtain
the coefficient of $J_0(\lambda)^{-1}$ of $S_{2n+2j-1,22}$ is a polynomial in $q_1,q_2,\cdots,q_m$ and $r_1,r_2,\cdots,r_m$, and its highest-order terms in $q_1,q_2,\cdots,q_m$ and $r_1,r_2,\cdots,r_m$ are determined by the coefficient of $J_0(\lambda)^{-1}$ of $S_{2,2}^{2n+2j-1}+\sum_{i=0}^{2n+2j-3}S_{2,2}^{i}S_{2,1}S_{1,2}S_{2,2}^{2n+2j-3-i}$, which is of order $2(j-1)$.
\end{proof}

While no general formula exists for arbitrary $j$ and $k$, recursive computation is feasible. Examples are given.

\begin{ex}
\begin{itemize}
\item[\rm(1).] Set $m=0, n=1$ as in Example \ref{exYV}. Let
$$
u=\left(\begin{array}{ccc}
  0 &-u_1 & 0  \\
 0 &0 & -u_1 \\
  0 &0 & 0
\end{array}\right)\in V_{0,1}.
$$
We have
$$
\left(\ln \tau_f\right)_{t_{1} t_{1}}=\frac{2u_1}{3}.
$$

\item[\rm(2).] Set $m=0, n=2$ as in Example \ref{exYV}. Let
$$
u=\left(\begin{array}{ccccc}
0 & 0 & 0 & -u_2 & 0 \\
0 & 0 & -u_1 & 0 & -u_2 \\
0 & 0 & 0 & -u_1 & 0  \\
0 & 0 & 0 &0 & 0 \\
0 & 0 & 0 &0 & 0
\end{array}\right)\in V_{0,2}.
$$
We have
$$
\begin{aligned}
&\left(\ln \tau_f\right)_{t_{1} t_{1}}=\frac{2u_1}{5},
&\left(\ln \tau_f\right)_{t_{1} t_{3}}=\frac{6 u_2}{5}+\frac{4 u_{1,x}}{5}-\frac{12u_1^2}{25}.
\end{aligned}
$$

\item[\rm(3).] Set $m=1, n=0$ as in Example \ref{exYV}.
Let
$$
u=\left(\begin{array}{ccccc}
0 & 0 & q_1  \\
0 & 0 & -r_1  \\
-r_1 & q_1 & 0
\end{array}\right)\in V_{1,0}.
$$
We have
$$
\begin{aligned}
&\left(\ln \tau_f\right)_{t_{1} t_{1}}=2 q_1 r_1,
&\left(\ln \tau_f\right)_{t_{1} t_{3}}=2q_1r_{1,xx}+2q_{1,xx}r_{1}+12 q_1^2r_1^2-2 q_{1,x}r_{1,x}.
\end{aligned}
$$

\item[\rm(4).] Set $m=1, n=1$ as in Example \ref{exYV}. Let
$$
u=\left(\begin{array}{ccccc}
0 & 0 & 0 & 0 & q_1 \\
0 & 0 & 0 & 0 & -r_1 \\
-r_1 & q_1 & 0 & -u_1 & 0  \\
0 & 0 & 0 &0 & -u_1 \\
0 & 0 & 0 &0 & 0
\end{array}\right)\in V_{1,1}.
$$
We have

{\small
$$
\begin{aligned}
&\left(\ln \tau_f\right)_{t_{1} t_{1}}=\frac{2u_1}{3},
\left(\ln \tau_f\right)_{t_{1} t_{3}}=2 q_1 r_1,\\
&\left(\ln \tau_f\right)_{t_{1} t_{5}}=\frac{40}{9}u_1q_1r_1-\frac{40}{81}u_1^3-\frac{5}{27}(u_{1,x})^2-\frac{8}{27}u_1u_{1,x}+\frac{1}{27}u_{1,xxxx}-3q_{1,xx}r_1-3q_1r_{1,xx}-\frac{8}{3}q_{1,x}r_{1,x}.
\end{aligned}
$$
}
\end{itemize}

\end{ex}

\section{\sc \bf Virasoro action for the $\hat{A}_{2 n}^{(2)}$-KdV hierarchy}

In this section we construct the Virasoro action on $\ln \tau_f$ for the $\hat{A}_{2 n}^{(2)}$-KdV hierarchy.

Let group $(G_{0,n})_{ \pm}$ and Lie algebra  $\left(\mathcal{G}_{0,n}\right)_{\pm}$ be as in section 2. The pair $\left(\left(\mathcal{G}_{0,n}\right)_{+}, \left(\mathcal{G}_{0,n}\right)_{-}\right)$ is a splitting of $\mathcal{G}_{0,n}$.
Let the vacuum sequence $J_0^{2i-1}$ be as in \eqref{vacseq} and $V(t)=\exp(\sum_{i=1}^{\infty} t_{2i-1}J_0^{2i-1})$. Given $f\in (G_{0,n})_-$, we can factor
\begin{gather}\label{ME}
V(t) f^{-1}=M^{-1} E, \quad \text { with } M \in (G_{0,n})_{-} \quad \text { and } \quad E \in (G_{0,n})_{+}.
\end{gather}
According to the conclusions of Terng and Uhlenbeck \cite{TerngUKTf2016,TerngUKTf20162}, we can derive that
\begin{gather*}
\delta_{l}(f)=-\left(\lambda^{l}\left(\lambda f_\lambda f^{-1}+f \Xi f^{-1}\right)\right)_{-} f
\end{gather*}
are Virasoro vector fields on $(G_{0,n})_{-}$ and the induced Virasoro vector fields on $\ln \tau_f$ are
\begin{gather}\label{tauV0}
\delta_{l}\left(\ln \tau_f\right) =\left\langle\lambda^{l} E\left(\lambda f_\lambda f^{-1}+f \Xi f^{-1}\right) E^{-1}, \lambda E_\lambda E^{-1}\right\rangle_0,
\end{gather}
where $\Xi=\frac{1}{2n+1} \operatorname{diag}(0,1, \ldots,2n)$.
Henceforth in this section we use the following notations:
$$
\begin{aligned}
\Gamma f  =\lambda f_\lambda+f \Xi.
\end{aligned}
$$

We can similarly obtain the following two lemmas in \cite{TerngUKTf2016,TerngUKTf20162}.

\begin{lemma} \label{lemMu61}
Let $M$ denote the reduced frame of the formal inverse scattering solution $u_f$, and $Q=M J_0 M^{-1}$, $P=(\Gamma M) M^{-1}$. Then we have
\begin{gather*}
E\left((\Gamma f) f^{-1}\right) E^{-1}+\lambda E_\lambda E^{-1}=(\Gamma M) M^{-1}+M \hat{\mathcal{J}} M^{-1},\\
\lambda\left(Q^j\right)_\lambda=\left[P, Q^j\right]+\frac{j}{2n+1} Q^j, \\
\operatorname{tr}\left(\lambda\left(Q^j\right)_\lambda Q^{2n+1-j}\right)=j \lambda, \quad 1 \leq j \leq 2n,\\
\left\langle P, \lambda^{2i-1}\right\rangle_0=\left\langle\lambda^{2i-1},(\Gamma f) f^{-1}\right\rangle_0,i=1,2,\cdots,
\end{gather*}
where $\hat{\mathcal{J}}=\frac{1}{2n+1} \sum_{i=1}^{\infty} (2i-1) t_{2i-1} J_0^{2i-1}$.
\end{lemma}
\begin{lemma}\label{lemST62}
Let $S, T \in \mathcal{G}_{0,n}$. Then
\begin{gather*}
\langle\lambda\partial_\lambda S, T\rangle_{0}+\langle S, \lambda\partial_\lambda T\rangle_{0}=0, \\
\langle\lambda\partial_\lambda S, T\rangle_{0}=\left\langle\lambda\partial_\lambda S_{+}, T\right\rangle_{0}+\left\langle\lambda\partial_\lambda S, T_{+}\right\rangle_{0}, \\
\left\langle\partial_\lambda(\lambda S)_{+}, \lambda\partial_\lambda T\right\rangle_{0}=\left\langle\lambda\partial_\lambda S_{+}, \partial_\lambda(\lambda T)\right\rangle_{0}+2\left\langle S_{+}, \lambda\partial_\lambda T\right\rangle_{0}.
\end{gather*}
\end{lemma}

Next we calculate the Virasoro vector fields on $\ln \tau_f$.

\begin{thm}\label{thmVirsoro}
The Virasoro vector fields on $\ln \tau_f$ are given by the following formulas:
$$
\begin{aligned}
\delta_{2k-1} (\ln \tau_f)&=0, k=0,1,2,\cdots, \\
\delta_0 (\ln \tau_f)= &\frac{1}{2n+1} \sum_{i =1}^{\infty} (2i-1) t_{2i-1}  (\ln \tau_f)_{t_{2i-1}}, \\
\delta_{2k} (\ln \tau_f)=& \frac{1}{2n+1} \sum_{i =1}^{\infty} (2i-1) t_{2i-1}  (\ln \tau_f)_{t_{2k(2n+1)+(2i-1)}}\\
&+\frac{1}{2n+1} \sum_{i=1}^{n}\sum_{j=1}^{k} (\ln \tau_f)_{t_{2(k-j)(2n+1)+(2i-1)}t_{2j(2n+1)-(2i-1)}}\\
&+
\frac{1}{2(2n+1)} \sum_{i=1}^{k(2n+1)} (\ln \tau_f)_{t_{2i-1}}(\ln \tau_f)_{t_{2k(2n+1)-(2i-1)}}\\
&+\left(\frac{1}{2(2n+1)}-\frac{1}{2}\right)c_{2k}(f),k=1,2,\cdots,
\end{aligned}
$$
where $c_{2k}(f)=\left\langle\lambda^{2k}\left((\Gamma f) f^{-1}\right),\left((\Gamma f) f^{-1}\right)\right\rangle_0$. Here we assume $(\ln \tau_f)_{t_{i(2n+1)}}=0,i=1,2,\cdots$.
\end{thm}

\begin{proof}
It follows from \eqref{tauV0} that
\begin{gather*}
\delta_{l} (\ln \tau_f)=\left\langle\lambda^{l} E(\Gamma f) f^{-1} E^{-1}, \lambda E_\lambda E^{-1}\right\rangle_0.
\end{gather*}
When $l=2k-1 (k=0,1,2,\cdots)$, the definition of Lie algebra $\mathcal{G}_{0,n}$ implies that $\delta_{l} (\ln \tau_f)=0$.

If $l=2k (k=0,1,2,\cdots)$, then by Lemma \ref{lemMu61},
we obtain
$$
\begin{aligned}
\delta_{2k} (\ln \tau_f) =&\frac{1}{2}\left\langle\lambda^{2k}\left((\Gamma M) M^{-1}\right),\left((\Gamma M) M^{-1}\right)\right\rangle_0+\frac{1}{2}\left\langle\lambda^{2k} \hat{\mathcal{J}},\hat{\mathcal{J}}\right\rangle_0\\
&+\left\langle\lambda^{2k}(\Gamma M) M^{-1}, M \hat{\mathcal{J}} M^{-1}\right\rangle_0-\frac{1}{2} c_{2k}(f).
\end{aligned}
$$
where $c_{2k}(f)=\left\langle\lambda^{2k}\left((\Gamma f) f^{-1}\right),\left((\Gamma f) f^{-1}\right)\right\rangle_0$.
Set
$$
\begin{aligned}
(\mathrm{I})&=\frac{1}{2}\left\langle\lambda^{2k}\left((\Gamma M) M^{-1}\right),\left((\Gamma M) M^{-1}\right)\right\rangle_0,\\
(\mathrm{II})&=\frac{1}{2}\left\langle\lambda^{2k} \hat{\mathcal{J}},\hat{\mathcal{J}}\right\rangle_0,\\
(\mathrm{III})&=\left\langle\lambda^{2k}(\Gamma M) M^{-1}, M \hat{\mathcal{J}} M^{-1}\right\rangle_0.
\end{aligned}
$$

First we evaluate $(\mathrm{I})$. If $k=0$, then a straightforward computation shows $(\mathrm{I})=\frac{1}{2}c_{0}(f)$.
If $k=1,2,\cdots$, we set $Q=M J_0 M^{-1}$ and $P=(\Gamma M) M^{-1}$, then
$$
\begin{aligned}
(\mathrm{I})=\frac{1}{2(2n+1)}\sum_{i=1}^{2n+1}\left\langle\lambda^{2k-2}Q^{2i-1}P,P Q^{2(2n+1)-(2i-1)}\right\rangle_0.
\end{aligned}
$$
By using
\begin{gather*}
\operatorname{tr}([P,\xi][P,\eta])=2\operatorname{tr}(P\xi P\eta)-\operatorname{tr}(P^2(\xi\eta+\eta\xi)),
\end{gather*}
we get
$$
\begin{aligned}
(\mathrm{I})=&\frac{1}{2(2n+1)}\sum_{i=1}^{2n+1}\left\langle\lambda^{2k-2}P Q^{2i-1},P Q^{2(2n+1)-(2i-1)}\right\rangle_0\\
&-\frac{1}{4(2n+1)}\sum_{i=1}^{2n+1}\left\langle\lambda^{2k-2}[P, Q^{2i-1}],[P, Q^{2(2n+1)-(2i-1)}]\right\rangle_0.
\end{aligned}
$$
Set
$$
\begin{aligned}
(\mathrm{I}_1)&=\frac{1}{2(2n+1)}\sum_{i=1}^{2n+1}\left\langle\lambda^{2k-2}P Q^{2i-1},P Q^{2(2n+1)-(2i-1)}\right\rangle_0,\nonumber\\
(\mathrm{I}_2)&=-\frac{1}{4(2n+1)}\sum_{i=1}^{2n+1}\left\langle\lambda^{2k-2}[P, Q^{2i-1}],[P, Q^{2(2n+1)-(2i-1)}]\right\rangle_0.\nonumber
\end{aligned}
$$
A routine calculation reveals that
\begin{gather*}
(\mathrm{I}_1)=\frac{1}{2(2n+1)}\sum_{i=1}^{2n+1}\sum_{j=-2}^{2k}\left(\left\langle P, Q^{2i-1}\right\rangle_{-j}\left\langle P, Q^{2k(2n+1)-(2i-1)}\right\rangle_{j}\right).
\end{gather*}
By a direct calculation, we obtain $\left\langle P, Q^{2i-1}\right\rangle_{2}=0$.
the definition of Lie algebra $\mathcal{G}_{0,n}$ yields that
\begin{gather*}
\left\langle P, Q^{2i-1}\right\rangle_{2j-1}=0,(j=0,1,2,\cdots,k).
\end{gather*}
Hence we have
$$
\begin{aligned}
(\mathrm{I}_1)&=\frac{1}{2(2n+1)}\sum_{i=1}^{2n+1}\sum_{j=0}^{k}\left(\left\langle P, Q^{2i-1}\right\rangle_{-2j}\left\langle P, Q^{2k(2n+1)-(2i-1)}\right\rangle_{2j}\right),\\
&=\frac{1}{2(2n+1)}\sum_{i=1}^{2n+1}\sum_{j=0}^{k}\left(\left\langle P, Q^{2j(2n+1)+(2i-1)}\right\rangle_{0}\left\langle P, Q^{2k(2n+1)-(2j(2n+1)+(2i-1))}\right\rangle_{0}\right),\\
&=\frac{1}{2(2n+1)}\sum_{i=1}^{k(2n+1)}\left(\left\langle P, Q^{2i-1}\right\rangle_{0}\left\langle P, Q^{2k(2n+1)-(2i-1))}\right\rangle_{0}\right).
\end{aligned}
$$
By Lemma \ref{lemMu61}, we have
\begin{gather*}
\sum_{i=1}^{k}\left\langle P, \lambda^{2i-1}\right\rangle_{0}\left\langle P, \lambda^{2k-(2i-1)}\right\rangle_{0}=c_{2k}(f).
\end{gather*}
Since
\begin{gather*}
\left\langle P, Q^{2i-1}\right\rangle_{0}=\left\langle \lambda M_\lambda M^{-1}+M \Xi M^{-1}, Q^{2i-1}\right\rangle_{0}=(\ln \tau_f)_{t_{2i-1}},
\end{gather*}
we get
\begin{gather*}
(\mathrm{I}_1)=\frac{1}{2(2n+1)} \sum_{i=1}^{k(2n+1)} (\ln \tau_f)_{t_{2i-1}}(\ln \tau_f)_{t_{2k(2n+1)-(2i-1)}}+c_{2k}(f).
\end{gather*}
From Lemma \ref{lemMu61}, it follows that
\begin{gather*}
(\mathrm{I}_2)=-\frac{1}{2(2n+1)}\sum_{i=1}^{n}\left\langle\lambda \partial_\lambda (Q^{2i-1}), \partial_\lambda (\lambda^{2k-1}Q^{2(2n+1)-(2i-1)})\right\rangle_{0}.
\end{gather*}
Then a straightforward computation shows that
\begin{gather*}
(\mathrm{I}_2)=-\frac{1}{2(2n+1)}\sum_{i=1}^{n}\left\langle\lambda \partial_\lambda (Q^{2i-1}), \partial_\lambda (\lambda^{2k-1}Q^{2(2n+1)-(2i-1)})_{+}\right\rangle_{0}.
\end{gather*}
From Lemma \ref{lemST62}, we conclude that
$$
\begin{aligned}
(\mathrm{I}_2)&=-\frac{1}{2n+1}\sum_{i=1}^{n}\sum_{j=1}^{k}\left\langle (\lambda^{2k-2j} Q^{2i-1})_+, \partial_\lambda (\lambda^{2j-2}Q^{2(2n+1)-(2i-1)})\right\rangle_{0}\\
&=\frac{1}{2n+1} \sum_{i=1}^{n}\sum_{j=1}^{k} (\ln \tau_f)_{t_{2(k-j)(2n+1)+(2i-1)}t_{2j(2n+1)-(2i-1)}}.
\end{aligned}
$$

Next we evaluate $(\mathrm{II})$. Computing directly, we find
$$
\begin{aligned}
(\mathrm{II})=\frac{1}{2}\left\langle \frac{1}{2n+1} \sum_{i=1}^{\infty} (2i-1) t_{2i-1}\lambda^{2k} J_0^{2i-1},\frac{1}{2n+1} \sum_{j=1}^{\infty} (2j-1) t_{2j-1} J_0^{2j-1}\right\rangle_0=0.
\end{aligned}
$$

Finally, we compute $(\mathrm{III})$. Then we have
$$
\begin{aligned}
(\mathrm{III})=\left\langle\lambda^{2k+1} M_\lambda M^{-1}, M \hat{\mathcal{J}} M^{-1}\right\rangle_0
+\left\langle\lambda^{2k}M \Xi M^{-1}, M \hat{\mathcal{J}} M^{-1}\right\rangle_0.
\end{aligned}
$$
Note that
$$
\begin{aligned}
\left\langle\lambda^{2k}M \Xi M^{-1}, M \hat{\mathcal{J}} M^{-1}\right\rangle_0=0.
\end{aligned}
$$
Thus,
$$
\begin{aligned}
(\mathrm{III})=\frac{1}{2n+1} \sum_{i =1}^{\infty} (2i-1) t_{2i-1}  (\ln \tau_f)_{t_{2k(2n+1)+(2i-1)}}.
\end{aligned}
$$
This completes the proof of the theorem.
\end{proof}

\section{\sc \bf Conclusions and Dicsussions}

In this paper, we construct the matrix constrained CKP hierarchy given by Definition \ref{def2V}. Theorem \ref{thmeqv3.5} shows that this hierarchy is equivalent to constrained CKP hierarchy. Then the Darboux transformations, permutability formula, scaling transformations and tau functions of the matrix constrained CKP hierarchy are given in Theorems \ref{thmDT4.7}, \ref{thmPF5.2}, \ref{thmST6.2} and \ref{thmtaufunction2}, respectively.
When $m=0$, the matrix constrained CKP hierarchy is reduced to the case of the $\hat{A}_{2n}^{(2)}$-KdV hierarchy was studied in \cite{TerngWuDt2023}. Theorem \ref{thmVirsoro} illustrates that the Virasoro vector fields on $\ln \tau_f$ for the $\hat{A}_{2 n}^{(2)}$-KdV hierarchy are given by partial differential operators.

If $q_i=r_i$, a reduction of the constrained CKP hierarchy \cite{LorisOr1999} is given as
\begin{eqnarray}\label{ccCKPL}
L^{2n+1}=(L^{2n+1})_{\geq0}+\sum_{i=1}^{m}\tilde{q}_{i}\partial^{-1}\tilde{q}_{i},\nonumber
\end{eqnarray}
where the $\tilde{q}_{i}$'s are eigenfunctions. It can be generated from the splitting of subalgebras of $\mathcal{G}_{m,n}$.
The involution of $G_{m,n}$ is denoted by $\zeta$, which is defined as
\begin{eqnarray}
\zeta(Y)=\left(\operatorname{diag}(I_{2m}, \beta_n) Y^t (\operatorname{diag}(I_{2m}, \beta_n))^{-1}\right)^{-1}, \;Y \in G_{m,n}.\nonumber
\end{eqnarray}
Then the induced involution $\theta$ on $\mathcal{G}_{m,n}$ is
\begin{eqnarray}
\zeta_*(y)=-\operatorname{diag}(I_{2m}, \beta_n) y^t (\operatorname{diag}(I_{2m}, \beta_n))^{-1}, \;y\in \mathcal{G}_{m,n}.\nonumber
\end{eqnarray}
Let
\begin{eqnarray}
\tilde{\mathcal{G}}_{m,n}=\left\{(\lambda)=\sum_{i \leq i_0} \xi_i \lambda^i \mid
\zeta_*(\xi(-\lambda))=\xi(\lambda),\;
 \xi_i \in \mathcal{G}_{m,n}\right\}.\nonumber
\end{eqnarray}
Similarly, $\left(\tilde{\mathcal{G}}_{m,n}\right)_{+}, \left(\tilde{\mathcal{G}}_{m,n}\right)_{-},\left(\tilde{\mathcal{G}}_{m,n}\right)_{[j]}, \tilde{G}_{m,n}$ and $(\tilde{G}_{m,n})_{ \pm}$ can be defined by Section 2.
We take a vacuum sequence of the splitting $\left(\left(\tilde{\mathcal{G}}_{m,n}\right)_{+}, \left(\tilde{\mathcal{G}}_{m,n}\right)_{-}\right)$  as
\begin{eqnarray}
\tilde{\mathcal{J}}=\left\{\tilde{J}^{2i-1} \mid \tilde{J}=\operatorname{diag}(0_{2m},e_{1,2n+1}\lambda+\sum_{j=1}^{2n}e_{j+1,j}),  i=1, 2, 3,\cdots\right\}\nonumber
\end{eqnarray}
in $(\tilde{G}_{m,n})_{+}$. Then we can naturally define the matrix form for reduction of the constrained CKP hierarchy defined by \eqref{ccCKPL} and get its Darboux transformations, permutability formula, scaling transformations and tau functions.

On the other hand, the constrained CKP hierarchy defined by \eqref{ccCKPL} can also be generated from a splitting of a Lie algebra.
The involution of $SL(m+2n+1,\mathbb{C})$ is denoted by $\gamma$, which is defined as
\begin{eqnarray}
\gamma(Y)=\left(\operatorname{diag}(I_{m}, \beta_n) Y^t (\operatorname{diag}(I_{m}, \beta_n))^{-1}\right)^{-1}, \;Y \in SL(m+2 n+1, \mathbb{C}).\nonumber
\end{eqnarray}
Then the induced involution $\gamma_*$ on $sl(m+2 n+1, \mathbb{C})$ is
\begin{eqnarray}
\gamma_*(y)=-\operatorname{diag}(I_{m}, \beta_n) x^t (\operatorname{diag}(I_{m}, \beta_n))^{-1}, \;y\in sl(m+2 n+1, \mathbb{C}).\nonumber
\end{eqnarray}
Let
\begin{eqnarray}
\hat{\mathcal{G}}_{m,n}=\left\{\xi(\lambda)=\sum_{i \leq i_0} \xi_i \lambda^i \mid
\gamma_*(\xi(-\lambda))=\xi(\lambda),\;
 \xi_i \in sl(m+2n+1, \mathbb{C})\right\}.\nonumber
\end{eqnarray}
Similarly, $ \hat{\mathcal{G}}_{m,n}$ and $(\tilde{G}_{m,n})_{ \pm}$ can be defined by Section 2.
We take a vacuum sequence of the splitting $\left(\left(\hat{\mathcal{G}}_{m,n}\right)_{+}, \left(\hat{\mathcal{G}}_{m,n}\right)_{-}\right)$  as
\begin{eqnarray}
\hat{\mathcal{J}}=\left\{\hat{J}^{2i-1} \mid \tilde{J}=\operatorname{diag}(0_{m},e_{1,2n+1}\lambda+\sum_{j=1}^{2n}e_{j+1,j}),  i=1, 2, 3,\cdots\right\}\nonumber
\end{eqnarray}
in $(\hat{\mathcal{G}}_{m,n})_{+}$. Then we can naturally define the matrix form for reduction of the constrained CKP hierarchy defined by \eqref{ccCKPL} and get its Darboux transformations, permutability formula, scaling transformations and tau functions.

In the following, it is algebraically shown that the $GD_{2n+1}$ hierarchy generated by
$$L=\partial^{2n+1}+\sum_{i=1}^{n}\left(\partial^{n+1-i}u_i\partial^{n-i}+\partial^{n-i}u_i\partial^{n+1-i}\right)$$
is a sub-hierarchy of the constrained CKP hierarchy generated by
$$L=\partial^{2n+1}+\sum_{i=1}^{n}\left(\partial^{n+1-i}u_i\partial^{n-i}+\partial^{n-i}u_i\partial^{n+1-i}\right)+\sum_{i=1}^{m}\left(q_{i}\partial^{-1}r_{i}+r_{i}\partial ^{-1}q_{i}\right).$$  
The involution of $G_{m,n}$ is denoted by $\varepsilon$, which is defined as
\begin{eqnarray}
\varepsilon(Y)=\left(\operatorname{diag}(-I_{2m}, \beta_n) Y^t (\operatorname{diag}(-I_{2m}, \beta_n))^{-1}\right)^{-1}, \;Y \in G_{m,n}.\nonumber
\end{eqnarray}
Then the induced involution $\varepsilon_*$ on $\mathcal{G}_{m,n}$ is
\begin{eqnarray}
\varepsilon_*(y)=-\operatorname{diag}(-I_{2m}, \beta_n) y^t (\operatorname{diag}(-I_{2m}, \beta_n))^{-1}, \;y\in \mathcal{G}_{m,n}.\nonumber
\end{eqnarray}
Let
\begin{eqnarray}
\tilde{\mathcal{G}}_{m,n}=\left\{\xi(\lambda)=\sum_{i \leq i_0} \xi_i \lambda^i \mid
\varepsilon_*(\xi(-\lambda))=\xi(\lambda),\;
 \xi_i \in \mathcal{G}_{m,n}\right\}.\nonumber
\end{eqnarray}
We take a vacuum sequence as
\begin{eqnarray}
\tilde{\mathcal{J}}=\left\{\tilde{J}^{2i-1} \mid \tilde{J}=\operatorname{diag}(0_{2m},e_{1,2n+1}\lambda+\sum_{j=1}^{2n}e_{j+1,j}),  i=1, 2, 3,\cdots\right\}.\nonumber
\end{eqnarray}
Then by using splitting theory, we also can define the $GD_{2n+1}$ hierarchy generated by
\begin{eqnarray}
L=\partial^{2n+1}+\sum_{i=1}^{n}\left(\partial^{n+1-i}u_i\partial^{n-i}+\partial^{n-i}u_i\partial^{n+1-i}\right),\nonumber
\end{eqnarray}
which is a sub-hierarchy of the constrained CKP hierarchy.

\section*{\bf Acknowledgements}

 The research of Tian, K. is supported by the National Natural Science Foundation of China under Grant No. 12171133. The research of Wu, Z. is supported by the National Natural Science Foundation of China under Grant No. 12271535 and No. 12431008.

\end{document}